\documentclass[12pt,reqno]{amsart}
\usepackage[a4paper,hmargin=2cm]{geometry}
\usepackage[T1]{fontenc}
\usepackage[utf8]{inputenc}
\usepackage{xcolor}
\usepackage{amsmath}
\usepackage{amsthm}
\usepackage{amssymb}
\usepackage{graphicx}
\usepackage[style=alphabetic,backend=biber]{biblatex}
\usepackage[shortlabels]{enumitem}
\usepackage{hyperref}
\usepackage{amsaddr}
\usepackage{tikz-cd} 
\usepackage{marginnote}
\def\showlabelsetlabel#1{\raise2ex\hbox{{\fontsize{2.5}{4}\selectfont #1}}}
\makeatother

\definecolor{notefontcolor}{rgb}{0.800781, 0.800781, 0.800781}
\definecolor{grey30}{rgb}{0.7,0.7,0.7}
\usepackage{ulem}
\numberwithin{equation}{section}
\theoremstyle{plain}
\usepackage{graphicx, color, svg}
\usepackage{siunitx}
\usepackage{float}
\usepackage{tikz}
\usepackage{caption}
\usetikzlibrary{calc}
\usetikzlibrary{decorations.pathreplacing}
\newcommand{\beq}{ \begin{equation}} 
\newcommand{\eeq}{\end{equation}} 
\newcommand{\bea}{\begin{aligned}}
\newcommand{\eea}{\end{aligned}}
\newcommand{\bdm}{\begin{displaymath}}
\newcommand{\edm}{\end{displaymath}}
\newcommand{\barr}{\begin{array}}
\newcommand{\earr}{\end{array}}
\newcommand{\ben}{\begin{enumerate}}
\newcommand{\een}{\end{enumerate}}
\newcommand{\bde}{\begin{description}}
\newcommand{\ede}{\end{description}}

\newtheorem{teor}{Theorem}
\newtheorem{prop}[teor]{Proposition} 
\newtheorem{lem}[teor]{Lemma}  
\newtheorem{cor}[teor]{Corollary}

\newtheorem{rem}[teor]{Remark}
\numberwithin{equation}{section}  
\numberwithin{teor}{section}
\newcommand{\R}{\mathbb{R}}

\newcommand{\E}{{\mathbb{E}}}

\DeclareMathOperator{\TAP}{TAP}
\newcommand\TP{relevant TAP}
\DeclareMathOperator{\argmax}{argmax}
\DeclareMathOperator{\On}{On}

\DeclareMathOperator{\Ann}{Ann}
\newcommand\mhattext{\hat{m}}
\newcommand\mhatdisplay{\mhattext}
\newcommand\qstar{q_*}

\newcommand\ITAP{I_{\TAP}}

\newcommand\ETAP{U}
\newcommand{\Eq}{E_{q}}
\newcommand{\qE}{q_{E}}
\newcommand{\EU}{E_{U}}
\newcommand{\qU}{q_{U}}
\newcommand{\UE}{U_{E}}

\newcommand{\Emin}{E_{\min}}
\newcommand\ETAPmax{\ETAP_{\max}}
\newcommand\ETAPmin{\ETAP_{\min}}
\newcommand\ETAPstar{\ETAP_*}
\newcommand\Estar{E_*}
\newcommand\NU{\mathcal{N}_N}

\newcommand\fTAP{f_{\TAP}}
\newcommand{\betatwosup}{\beta/\tilde{\beta}_c}
\newcommand\Uminbeta{U_{\min}}
\newcommand\HTAPg{H_{\TAP}^g}
\newcommand\HNg{g}
\newcommand\gTAP{$g$-TAP}
\usepackage[shortlabels]{enumitem}
        \setlist[enumerate, 1]{1)}

\begin{document}

\title[Phase diagram for $p$-spin TAP energy]{Phase diagram for the TAP energy of the $p$-spin spherical mean field spin glass model}

\author{David Belius, Marius A. Schmidt}
\email{david.belius@cantab.net}
\email{M.Schmidt@mathematik.uni-frankfurt.de}
\address{Department of Mathematics and Computer Science, University of Basel, Switzerland}
\thanks{Both authors supported by SNSF grant 176918.}

\maketitle

\begin{abstract} 
We solve the Thouless-Anderson-Palmer (TAP) variational principle associated to the spherical pure $p$-spin mean field spin glass Hamiltonian and present a detailed phase diagram.

In the high temperature phase the maximum of variational principle is the annealed free energy of the model. In the low temperature phase the maximum, for which we give a formula, is strictly smaller.

The high temperature phase consists of three subphases. (1) In the first phase $m=0$ is the unique relevant TAP maximizer.  (2) In the second phase there are exponentially many TAP maximizers, but $m=0$ remains dominant. (3) In the third phase, after the so called dynamic phase transition, $m=0$ is no longer a relevant TAP maximizer, and exponentially many non-zero relevant TAP solutions add up to give the annealed free energy.

Finally in the low temperature phase a subexponential number of TAP maximizers of near-maximal TAP energy dominate.
\end{abstract}

\section{Introduction}
In the physics literature on mean field spin glass models such as the Sherrington-Kirkpatrick model the Thouless-Anderson-Palmer (TAP) equations and TAP energy play an important role \cite{TAP1, mezard1987spin,brayMetastableStatesSpin1980,dedominicisWeightedAveragesOrder1983,grossSimplestSpinGlass1984,origin1,origin2,kurchanBarriersMetastableStates1993,cugliandoloAnalyticalSolutionOffequilibrium1993,crisantiThoulessAndersonPalmerApproachSpherical1995,barratDynamicsMetastableStates1996,cavagnaFormalEquivalenceTAP2003}. In mathematics their meaning and implications are an active area of study \cite[Section 1.7]{talagrand2010mean}, \cite{chatterjeeSpinGlassesStein2010,bolthausenIterativeConstructionSolutions2014,subagGeometryGibbsMeasure2017,chenTAPFreeEnergy2018,auffingerSpinDistributionsGeneric2019,DavidNicola,auffingerThoulessAndersonPalmer2019,chenGeneralizedTAPFree2021,chenGeneralizedTAPFree2022,subagFreeEnergyLandscapes2020,arousShatteringMetastabilitySpin2021,benarousGeometryTemperatureChaos2020,zhoufanTAPFreeEnergy2021,subagFreeEnergySpherical2021,brenneckeNoteReplicaSymmetric2021,TAPHTUB}. One basic idea is that the TAP energy encodes important information about the free energy and Gibbs measure of the model. In particular, the free energy should be given by a TAP variational principle. In this article we give a detailed phase diagram for the TAP variational principle of the pure $p-$spin spherical spin glass Hamiltonian.

Let $S_{N-1} \subset \mathbb{R}^N$ be the $N-1$-dimensional unit sphere and $B_N \subset {\mathbb{R}}^N$ the $N$-dimensional unit ball. Next for any power series $\xi(x)=\sum_{p\ge 1}a_p x^p$ with non negative coefficients $a_p\ge0$ satisfying $\xi(1)<\infty$ let the Hamiltonian $H_N$ be a centered Gaussian field on $B_{N}$ with covariance
\begin{equation}\label{eq: Hamiltonian covar}
    \E[H_{N}(\sigma) H_N(\tau)] = N \xi(\sigma \cdot \tau),\quad \sigma,\tau \in B_{N}.
\end{equation}
The free energy of the model is given by
\begin{equation}\label{eq: FE}
f_N \left(\beta \right) = \frac{1}{N} \log \int_{S_{N-1}} \exp\left( \beta H_N(\sigma) \right) d\sigma ,
\end{equation}
for an inverse temperature $\beta\ge0$.
As a step towards computing the free energy
the {\emph{TAP energy}} $H_{\TAP}$ has been introduced (for the standard SK model with $\pm1$ spins by the eponymous authors in \cite{TAP1}, and for the present model with spherical spins in \cite{kurchanBarriersMetastableStates1993}). It is given by
\beq  \label{eq:HTAPdef}
H_{\TAP}(m) = \beta H_N(m) +  \frac{N}{2}\log(1-|m|^2) + N \On(|m|^2), m\in B_N^{o}, \eeq
where $B_N^{o}$ is the open unit ball in $\mathbb{R}^n$
and the third term is the so called Onsager correction
\begin{equation}\label{eq: Onsager Def}
\On(q) = 
\frac{\beta^{2}}{2}\left(\xi(1)-\xi'(q)(1-q)-\xi(q)\right).
\end{equation}
A heuristic derivation of the TAP free energy illustrating the connection with the free energy
is given in Section \ref{sec:heuristic}. Only the $m$ that satisfy certain conditions are ``physical'' and relevant for the free energy.
In the physics literature it is widely accepted that to be a relevant, $m$ must be a local maximum of $H_{\TAP}$ and it must satisfy \emph{Plefka's condition}, which reads $\beta_2(|m|^2) \le \frac{1}{\sqrt 2}$,
where
\begin{equation}\label{eq: beta2 def}.
    \beta_2(q) = \beta \sqrt{ \frac{\xi''(q)}{2}} (1-q) .
\end{equation} 
In this paper we replace Plefka's condition by a slightly stronger condition, namely $|m|^2 \in D_\beta$ where
\beq \label{eq: our plefka} D_\beta = \{q\in [0,1]:  A(q,\beta) \le 0 \} ,\eeq
and
\begin{equation}\label{eq: Adef}
    A\left(q,\beta\right)=\sup_{r\in\left(0,1\right)}\left( \beta^{2}\frac{\xi'\left(q+r\left(1-q\right)\right)(1-q)-\xi'\left(q\right)\left(1-q\right)}{r}-\frac{1}{1-r}\right).
\end{equation}
In Lemma \ref{lem:ourplefkaprofs} we show that $|m|^2 \in D_\beta$ implies that Plefka's condition $\beta_2(|m|^2) \le \frac{1}{\sqrt 2}$ is satisfied. The opposite implication is however not true. Below we further comment on this condition, which is needed to obtain a coherent phase diagram.
A critical point of $H_{\TAP}$ is called a {\emph{TAP solution}}. We refer to $m$'s which are local maxima of $H_{\TAP}$ that satisfy $|m|^2 \in D_\beta$ as \emph{\TP\ solutions} and let the \emph{complexity} of \TP\ solutions of a certain energy and certain magnitude be given by the exponential rate
\beq \label{eq: def TAP comp}
\ITAP(\ETAP)
=
\lim_{\varepsilon\downarrow 0}
\lim_{N\rightarrow\infty}
\frac{1}{N}\log 
\left| \left\{
    m \in B^{o}_{N}: 
    \begin{array}{c}
        m \mbox{ loc. max. of } H_{\TAP},  |m|^2 \in D_\beta\\
        \frac{1}{N} H_{\TAP}(m) \in [U-\varepsilon,U+\varepsilon]
    \end{array}
\right\} \right|
\in
\{-\infty\}\cup[0,\infty),
\eeq
assuming the limits exist. Defining the \emph{total TAP free energy} by
\begin{equation}\label{eq:varprinc}
\fTAP(\beta) = \sup_{\ETAP\in \R} \{ \ETAP+ \ITAP(\ETAP) \}.
\end{equation}
a basic idea of the TAP ansatz is the claim that for all $\beta \ge 0$
\beq \label{eq:conj} \lim_{N\to\infty} f_N(\beta) = \fTAP(\beta).\eeq
A heuristic argument for this claim is given in Section \ref{sec:heuristic}. Combining the present paper with \cite{subagFreeEnergySpherical2021} it can be verified ``aposteriori'' (see below for a more detailed discussion). A direct proof of \eqref{eq:conj} is the subject of active research but is beyond the scope of this paper. However the claim \eqref{eq:conj} motivates the study of the variational principle \eqref{eq:varprinc}, and in this article we do so for the pure $p$-spin models, where 
\begin{equation}\label{eq: xi def}
\xi(x) = x^p\text{ for }p\ge3.
\end{equation}
We compute $\fTAP(\beta)$ and $\ITAP$ for all $\beta$, and give a detailed phase diagram characterizing the maximizers for different $\beta$.

We are able to compute the TAP complexity $\ITAP$ since for pure $p$-spin models the Hamiltonian $H_N$ is $p$-homogenous and therefore each TAP local maximum in $B_N$ corresponds to a local maximum of the Hamiltonian $H_N(\sigma)$ on $S_{N-1}$ (in the spherical metric). The complexity of critical points of $H_N(\sigma)$ has been determined by \cite{purefstmom,subagpure}. Their results imply  that if $\mathcal{M}(\cdot)$ is the number of local maxima of $H_N$ on $S_{N-1}$ with $\frac{H_N(\sigma)}{N} \in \cdot$ we have 
\begin{equation}
\lim_{\varepsilon \downarrow 0 } \lim_{N\to\infty} \frac{1}{N}\log \mathcal{M}([E-\varepsilon,E+\varepsilon]) =  I\left(E\right)\label{eq: loc max quenched}
\end{equation}
in probability, where $I$ is a function satisfying
\begin{equation}
I\left(E\right)=\begin{cases}
>0 & \text{if } E \in [E_\infty,E_0)\\
0 & \text{ if } E=E_0\\
-\infty & \text{ otherwise },
\end{cases}
\end{equation}
for 
\begin{equation}\label{eq: Einf def}
    E_\infty = \frac{2\sqrt{p-1}}{\sqrt{p}},
\end{equation}
and where $E_0>E_\infty$ is the limiting ground state energy of the Hamiltonian $H_N$, i.e.
\begin{equation}\label{eq: ground state}
    \sup_{\sigma \in S_{N-1}} H_N(\sigma) = N E_0 + o(N).
\end{equation}
The precise definitions of $I$ and $E_0$ are given in Section \ref{sec: crit point complexity}.

We now state our results for each phase in the phase diagram. The critical temperatures are given in terms of $E_\infty$ and $E_0$. Our first result shows that the energy surface $H_{\TAP}(m)$ and $\ITAP$ undergoes a phase transition at the {\emph{complexity threshold}} $\beta_c$. For $\beta<\beta_c$ there is only one \TP\ solution (and even only one TAP solution), namely $m=0$. For $\beta>\beta_c$ there are exponentially many \TP\ solutions.  To define $\beta_c$ we let
\begin{equation}\label{eq: beta tilde c def}
\tilde{\beta}_{c}=\frac{1}{2}\sqrt{\frac{p^{p-1}}{(p-1)(p-2)^{p-2}}},
\end{equation}
as well as
\begin{equation}\label{eq: rbardef}
\bar{r}=\frac{E_{0}}{E_{\infty}} - \sqrt{\left(\frac{E_{0}}{E_{\infty}}\right)^{2}-1},
\end{equation}
(note that $\overline{r}<1$) and set
\begin{equation}\label{eq: betacdef}
\beta_c = \bar{r}\tilde{\beta}_c.
\end{equation}
We then have the following.
\begin{teor}[Complexity threshold]\label{thm:comp}
For $\beta < \beta_c$ there are no \TP\ solutions except $m=0$, i.e.
\begin{equation}\label{eq:onlyzerosol}
\lim_{N \to \infty} P\left( m \text{ is a local maximum of }H_{\TAP}\text{ s.t. } |m|^2\in D_\beta \iff m=0 \right) = 1.
\end{equation}
In fact, a fortiori,
\begin{equation}\label{eq:onlyzerosol no plefka}
\lim_{N \to \infty} P\left( m\in B_N \text{ is a critical point of }H_{\TAP} \iff m=0 \right) = 1.
\end{equation}
In particular
\beq \label{eq:Itapbeforebetac}\ITAP(\ETAP) =
\begin{cases}
0 & \text{ if } \ETAP=H_{\TAP}(0)=\frac{\beta^2}{2},\\
-\infty & otherwise,
\end{cases}
\eeq
and
\begin{equation}\label{eq: FE very high temp}
    \fTAP(\beta) = \frac{\beta^2}{2}.
\end{equation}
For $\beta > \beta_c$ there are exponentially many \TP\ solutions, i.e.
\begin{equation}\label{eq: exp many TAP sol}
    \sup_{\ETAP} \ITAP(\ETAP)  > 0.
\end{equation}
\end{teor}

The next theorem describes $\ITAP$ in terms of $I$ when $\beta \ge \beta_c$.  To formulate the result let
\beq \label{Ndef} \NU(\mathcal{U},\mathcal{V},Q) = \left|\{m \in B_N: m \mbox{ loc. max. of } H_{\TAP}, \frac{1}{N} H_{\TAP}(m) \in \mathcal{U} , \frac{1}{N}H_N(m) \in \mathcal{V}, |m|^2 \in Q\}\right|,
\eeq
be the number of TAP solutions with given energy, given energy of the Hamiltonian $H_N(m)$ and given squared magnitude and the extended complexity
\begin{equation}\label{eq: ITAP extended}
I_{\TAP}(U,V,q) = \lim_{\varepsilon \downarrow 0 } \lim_{N\to\infty} \frac{1}{N} \log \mathcal{N}_N( [U-\varepsilon, U+\varepsilon], [V-\varepsilon,V+\varepsilon],[q-\varepsilon,q+\varepsilon]).
\end{equation}
That this limit exists is part of our results. Note that
$$ I_{\TAP}(U) = \sup_{V\in \mathbb{R},q\in D_\beta} I_{\TAP}(U,V,q).$$
For each $V,q$ there can be at most one $U$ such that $I_{\TAP}(U,V,q)\ne -\infty$, since $H_{\TAP}(m)$ is a function of $H_N(m)$ and $|m|^2$ (see \eqref{eq:HTAPdef}). For each $U$ there can in principle be several $(V,q)$ such that $I_{\TAP}(U,V,q) \ne -\infty$, but our analysis implies that this is never the case.
\begin{teor}\label{thm: TAP comp intro}
    For $\beta \ge \beta_c$ there exist $U_{\min},U_{\max}\in \mathbb{R}$ such that $U_{\min}<U_{\max}\le\frac{\beta^2}{2}$, with equality in the latter inequality only if $\beta=\beta_d$, and functions $q_U: [U_{\min},\infty) \to D_\beta\cap[ \frac{p-2}{p},1]$ and $E_U: [U_{\min},\infty) \to [E_{\infty},\infty)$ such that
    $$
    I_{{\rm TAP}}\left(U,V,q\right)=\begin{cases}
I\left(E\right) & \text{ if }U_{\min}\le U < U_{\max}\text{ and }q=q_U(U), E=E_U(U)\\
0 & \text{ if }U = U_{\max}\text{ and }q=q_U(U_{\max}), E=E_U(U_{\max})\\
0 & \text{\,if }U=\frac{\beta^{2}}{2}, V=0, q=0\text{ and }\beta\le\beta_{d}\\
-\infty & \text{otherwise}.
\end{cases}
$$
The functions $q_U,E_U$ are strictly increasing, and $E_U(U_{\max})=E_0$.
\end{teor}
A more complete but lengthier specification of $I_{\TAP}$ when $\beta \ge \beta_c$ is given in Section \ref{computeI}, and formulas for $U_{\min}$ and $U_{\max}$ are given in Section \ref{sec: umax umin}. Note that the theorem implies that when $\beta \ge \beta_c$
\begin{equation}
    \label{eq: def E TAP max}
    \frac{1}{N}\sup_{m\ne0\text{ is \TP\ sol.}}H_{\TAP}\left(m\right) \to U_{\max},
\end{equation}
\begin{equation}
\label{eq: def E TAP min}
    \frac{1}{N}\inf_{m\ne0\text{ is \TP\ sol.}}H_{\TAP}\left(m\right) \to U_{\min}.
\end{equation}
We exhibit two further phase transitions of \eqref{eq:varprinc} at the inverse temperatures
\begin{equation}\label{eq: beta d def}
    \beta_d = \sqrt{\frac{(p-1)^{p-1}}{p(p-2)^{p-2}}}
\end{equation}
corresponding to the dynamic phase transition of the spin glass model and
 \beq \label{betasdef}\beta_s = \sqrt{\frac{(p-1)^{p-1}}{p \bar{r}^2(p-1-\bar{r}^2)^{p-2}}}
 \eeq
which corresponds to the static phase transition. Using $\bar{r}<1$ one can check that indeed 
\begin{equation}\label{eq: order of betas}
    \beta_c < \tilde{\beta}_c < \beta_d < \beta_s \text{ for all }p\ge3.
\end{equation}
To formulate the results we let
\beq \label{eq: def E TAP star}
(\ETAPstar,V_*,q_*) = \underset{ U\in\R, V\in\R, q \in D_\beta}{\argmax}\left\{ \ETAP^{}+\ITAP(\ETAP,V,q)\right\}^{},
\eeq 
when the $\argmax$ is well-defined. When it is well-defined, $U_*$ is the maximizer in \eqref{eq:varprinc}, and $q_*$ is the  squared magnitude and $V_*$ the Hamiltonian energy of TAP solutions with energy $U_*$.  Note that Theorem \ref{thm:comp} implies that $\qstar = 0$, $V_*=0$ and $\ETAPstar = H_{\TAP}(0) = \frac{\beta^2}{2}$  for $\beta < \beta_c$.
The next theorem shows that while there are exponentially many \TP\ solutions for $\beta \in (\beta_c,\beta_d)$, the behavior of $\qstar,V_*$ and $\ETAPstar$ remains the same up to $\beta_d$, i.e. the maximizer in the variational principle \eqref{eq:varprinc} still corresponds to the \TP\ solution $m=0$. 

\begin{teor}[Phase of static and dynamic high temperature]\label{thm:sndthrm}If $\beta \le \beta_d$ then \eqref{eq: def E TAP star} is well defined and
\begin{equation}
\label{eq: qstar Ustar entropy high temp}
\mbox{(a) }\qstar = 0,
\quad \mbox{(b) }V_*=0\quad \mbox{(c) }\ETAPstar = H_{\TAP} (0) = \frac{\beta^2}{2} > \ETAPmax
\quad
\mbox{(d) }\ITAP(\ETAPstar,V_*,q_*) = 0,
\eeq
and therefore
\begin{equation}\label{eq: FE very high temp 2}
\fTAP(\beta) = \frac{\beta^2}{2}>U_{\max},
\end{equation}
(where we set $U_{\max}=-\infty$ for $\beta<\beta_c$).
\end{teor}
For $\beta >\beta_d$, we no longer have $0\in D_\beta$, so $m=0$ is no longer a \TP\ solution according to our definition.
However our next result shows that in $(\beta_d,\beta_s)$, it remains true that $\fTAP(\beta) = \frac{\beta^2}{2}$ but this value is achieved in a different way, i.e. $\qstar$, $V_*$, $\ETAPstar$, $\ITAP(\ETAPstar,V_*,q_*)$ are all given by different formulas.
Let
\begin{equation}\label{eq: hTAP def} 
h_{\TAP}(E,q) = \beta E + \frac{1}{2}\log(1-q) + \On(q),
\end{equation}
so that
\beq \label{eq: HTAP from hTAP}
H_{\TAP}(m) = \frac{1}{N} h_{\TAP}\left(\frac{H_N(m)}{N}, |m|^2\right).
\eeq
\begin{teor}[Phase of dynamic low temperature, static high temperature]\label{thm:dynLTstatHT}
For $\beta \in (\beta_d,\beta_s)$ the quantity \eqref{eq: def E TAP star} is well defined and it holds that $\qstar$ is the unique solution of 
\begin{equation}\label{eq: qstar dyn low temp}  (1 - q) q^{p-2} = \frac{1}{p\beta^2} \text{ in } \left[\frac{p-2}{p-1},1\right).
\end{equation}
Furthermore 
\beq \label{eq: Estar dyn LT stat HT}
V_* = q_*^{p/2}E_* \text{ for } E_*= \frac{E_{\infty}}{2}\left(\frac{1}{\sqrt{2}\beta_2(\qstar)}+\sqrt{2}\beta_2(\qstar)\right),
\eeq
as well as
\beq \label{eq: Ustar dyn LT stat HT}
\ETAPstar = h_{\TAP}(V_{*},\qstar) \eeq
hold. Additionally
\begin{equation}\label{eq: entropy dyn low temp}
\begin{array}{l}
\ITAP(\ETAPstar,V_*,q_*)=I(\Estar)
=-\frac{1}{2}\log(1-\qstar)-\frac{\qstar}{2}-\frac{\qstar^{2}}{2p(1-\qstar)}>0
\end{array}
\end{equation}
and
\begin{equation}\label{eq: FE dyn high temp}
\fTAP(\beta) = \frac{\beta^2}{2}.
\end{equation}
Lastly
\begin{equation}\label{eq:EstarnotGS}
   \ETAPmin < \ETAPstar < \ETAPmax.
\end{equation}
\end{teor}
The inequality \eqref{eq:EstarnotGS} shows that $\ETAPstar$ is not the maximum TAP energy in this phase. Furthermore \eqref{eq: entropy dyn low temp} shows that the value of $\fTAP(\beta)$ comes from the contribution of exponentially many \TP\ solutions $m$ of TAP energy $N\ETAPstar + o(N)$. This is in contrast to the static low temperature phase we describe next. Indeed after $\beta_s$, it is no longer true that $\fTAP(\beta) = \beta^2/2$. However once again $\ITAP(\ETAPmax) = 0$, signifying that the maximizer of \eqref{eq:varprinc} now corresponds to subexponentially many $m$ such that $H_{\TAP}(m) = \ETAPmax N + o(N)$  and $H_N(m) = E_0 N + o(N)$. It also gives a formula for $\fTAP(\beta)$, i.e. for the free energy in low temperature. 

\begin{teor}[Static and dynamic low temperature phase]\label{thm:fstthrm} For {$\beta > \beta_s$} we have that \eqref{eq: def E TAP star} is well-defined and $\qstar$ is the unique solution of 
\begin{equation}\label{eq: qstar low temp} \beta (1-q)q^{\frac{p-2}{2}}   = \frac{\bar{r}}{\sqrt{p(p-1)}} \text{ in } \left[\frac{p-2}{p},1\right).\end{equation}
Also 
\begin{equation}\label{eq: Vstar Ustar ITAPstar low temp}
\begin{array}{lll}
\text{(a) }V_* = q_*^{p/2}E_0  & \text{(b) }\ETAPstar = h_{\TAP}( V_*,\qstar) = \ETAPmax &
 \text{(c) }\ITAP(\ETAPstar,V_*,q_*) = 0.
\end{array}
\end{equation}
Finally $\fTAP(\beta)$ can be expressed in the various ways
\begin{equation}\label{eq: fTAP static low temp}
\begin{array}{ccl}
\fTAP(\beta) & = & \ETAPmax\\
&=&h_{\TAP}\left(V_*,\qstar\right)\\
& = & \displaystyle{\sup_{q\ge\frac{p-2}{p}:\sqrt{2}\beta_{2}(q)\le1}}h_{\TAP}(q^{p/2}E_{0},q)\\
 & = & \frac{\beta^{2}}{2}+\frac{1}{2}\log\left(1-q_{*}\right)+\frac{2}{p}\frac{q_{*}}{1-q_{*}}\frac{E_{0}}{E_{\infty}}\bar{r}-\frac{1}{2\left(p-1\right)}\left(1+\frac{1}{p}\frac{q_{*}}{1-q_{*}}\right)\frac{q_{*}}{1-q_{*}}\bar{r}^{2}
\end{array}
\end{equation}
and
\begin{equation}\label{eq: FE low temp}
\fTAP(\beta) < \frac{\beta^2}{2}.
\end{equation}
\end{teor}
The main points of the above results are summarized in the phase diagram Figure \ref{fig:phase diagram}. The critical temperatures appearing are summarized in Table \ref{tab:crit temps}.
All theorems follow from elementary but non-trivial calculations involving the complexity $I$ and the condition \eqref{eq: our plefka}.

\newpage

\begin{figure}[H]
 \centering
    \usetikzlibrary{automata,positioning}

\begin{tikzpicture}[yscale=0.5, xscale=0.8,every text node part/.style={align=center}]

\def\ymult{1.6}
\def\ytop{5*\ymult}
\def\yftap{\ytop}
\def\yzerotapsol{3*\ymult}
\def\yqstar{2*\ymult}
\def\yitapstar{1*\ymult}
\def\yitapnonzero{4*\ymult}
\def\ycurlybraces{-1*\ymult}

\input{phase_diagram/phase_diagram_axis}


\node (itapnonzerostart) at ({\xstart},\yitapnonzero)  {b)};
\node (itapnonzero1) at ({\xstart + (\xstart+\xbetac)/2},\yitapnonzero)  {$m=0$ only\\ \TP\ sol.};
\draw[-|, thick] (itapnonzero1.east)--(\xbetac,\yitapnonzero);
\draw[-, thick] (itapnonzero1.west)--(itapnonzerostart);
\node (itapnonzero2) at ({(\xbetac+\xend)/2},\yitapnonzero)  {Exponentially many \TP\ sol's};
\draw[-|, thick] (itapnonzero2.west)--(\xbetac,\yitapnonzero);
\draw[-, thick] (itapnonzero2.east)--(\xend,\yitapnonzero);


\node (zerotapsol1start) at ({\xstart},\yzerotapsol)  {c)};
\node (zerotapsol1) at ({\xstart + (\xstart+\xbetad)/2},\yzerotapsol)  {$m=0$ is \TP\ sol.};
\draw[-|, thick] (zerotapsol1.east)--(\xbetad,\yzerotapsol);
\draw[-, thick] (zerotapsol1.west)--(zerotapsol1start);
\node (zerotapsol2) at ({(\xbetad+\xend)/2},\yzerotapsol)  {$m=0$ not \TP\ sol.};
\draw[-|, thick] (zerotapsol2.west)--(\xbetad,\yzerotapsol);
\draw[-, thick] (zerotapsol2.east)--(\xend,\yzerotapsol);



\node (fTAP1start) at ({\xstart},\yftap)  {a)};
\node (fTAP1) at ({\xstart + (\xstart+\xbetas)/2},\yftap)  {$\fTAP=\frac{\beta^2}{2}$};
\draw[-|, thick] (fTAP1.east)--(\xbetas,\yftap);
\draw[-, thick] (fTAP1.west)--(fTAP1start);
\node (fTAP2) at ({(\xbetas+\xend)/2},\yftap)  {$\fTAP<\frac{\beta^2}{2}$};
\draw[-|, thick] (fTAP2.west)--(\xbetas,\yftap);
\draw[-, thick] (fTAP2.east)--(\xend,\yftap);


\node (qstar1start) at ({\xstart},\yqstar)  {d)};
\node (qstar1) at ({\xstart + (\xstart+\xbetas)/2},\yqstar)  {$q_{*}=0$};
\draw[-|, thick] (qstar1.east)--(\xbetad,\yqstar);
\draw[-, thick] (qstar1.west)--(qstar1start);
\node (qstar2) at ({(\xbetad+\xend)/2},\yqstar)  {$q_{*}>0$};
\draw[-|, thick] (qstar2.west)--(\xbetad,\yqstar);
\draw[-, thick] (qstar2.east)--(\xend,\yqstar);


\node (itapstar1start) at ({\xstart},\yitapstar)  {e)};
\node (itapstar1) at ({\xstart + (\xstart+\xbetad)/2},\yitapstar)  {$I_{\rm{TAP},*}=0$};
\draw[-|, thick] (itapstar1.east)--(\xbetad,\yitapstar);
\draw[-, thick] (itapstar1.west)--(itapstar1start);
            \node (itapstar2) at ({(\xbetad+\xbetas)/2},\yitapstar)  {$I_{\rm{TAP},*}>0$};
\draw[-|, thick] (itapstar2.west)--(\xbetad,\yitapstar);
\draw[-|, thick] (itapstar2.east)--(\xbetas,\yitapstar);
\node (itapstar3) at ({(\xbetas+\xend)/2},\yitapstar)  {$I_{\rm{TAP},*}=0$};
\draw[-|, thick] (itapstar3.west)--(\xbetas,\yitapstar);
\draw[-, thick] (itapstar3.east)--(\xend,\yitapstar);


\draw [decorate,decoration={mirror,brace,amplitude=10pt}]
(\xstart,\ycurlybraces) -- (\xbetad,\ycurlybraces) node [black,midway,below=10pt,align=left] {Stat. \& Dyn. HT};

\draw [decorate,decoration={mirror,brace,amplitude=10pt}]
(\xbetad,\ycurlybraces)--(\xbetas,\ycurlybraces) node  [black,midway,below=10pt,align=left]  {Stat. HT,\\ Dyn. LT};

\draw [decorate,decoration={mirror,brace,amplitude=10pt}]
(\xbetas,\ycurlybraces)--(\xend,\ycurlybraces) node  [black,midway,below=10pt,align=left]  {Stat. \& Dyn. LT};

\end{tikzpicture}
 \caption{
 Phase diagram of TAP variational principle \eqref{eq:varprinc}, \eqref{eq: def E TAP star}.\newline 
 a) TAP free energy $\fTAP(\beta)$ at high and low temperature, see \eqref{eq: FE very high temp}, \eqref{eq: FE very high temp 2}, \eqref{eq: FE dyn high temp}, \eqref{eq: FE low temp}.\newline
 b) Complexity transition, see Theorem \ref{thm:comp}.\newline
 c) Whether $m=0$ is \TP\ solution, see \eqref{eq: our plefka}. \newline
 d) Magnitude squared of \TP\ solutions maximizing TAP variational principle, see \eqref{eq: qstar Ustar entropy high temp} (a), \eqref{eq: qstar dyn low temp}, \eqref{eq: qstar low temp}.\newline
 e) Entropy $I_{TAP,*}=\ITAP(\ETAPstar,V_*,\qstar)$ of \TP\ solutions maximizing the TAP variational principle, see \eqref{eq: qstar Ustar entropy high temp} (c), \eqref{eq: entropy dyn low temp}, \eqref{eq: Vstar Ustar ITAPstar low temp} (c).
 }
 \label{fig:phase diagram}
\end{figure}

\begin{table}[H]
  \def\arraystretch{1.9}
    \begin{center}
    \begin{tabular}{| c | c | c | c |}
    \hline
    Notation & Formula  & Value ($p=3$) & Description\\
    \hline
    \hline
    $\beta_c$ & $\tilde{\beta}_c \bar{r} $ & 0.89372 & Complexity transition\\ 
    \hline
    $\tilde{\beta}_c$ & $ \frac{1}{2}\sqrt{\frac{p^{p-1}}{(p-1)(p-2)^{p-2}}} $ & 1.06066 & TAP loc. max. start existing $\forall$ $ E\in[E_\infty,E_0]$ \\
    \hline
    $\beta_d$  & $\sqrt{\frac{(p-1)^{p-1}}{p(p-2)^{p-2}}}$ & 1.15470 & Dynamical phase transition \\
    \hline
    $\beta_s$  & $\sqrt{\frac{(p-1)^{p-1}}{p \bar{r}^2(p-1-\bar{r}^2)^{p-2}}}$ & 1.20656 & Static phase transition\\
    \hline
    \end{tabular}
\end{center}        
  \captionof{table}{Critical temperatures; $\bar{r}$ is defined in \eqref{eq: rbardef}. The description for $\tilde{\beta}_c$ is explained in Remark \ref{rem: ITAP thm rems}. \label{tab:crit temps}}
\end{table}

\subsection{The condition \texorpdfstring{$|m|^2 \in D_\beta$}{m in D beta}}
The condition is motivated by the heuristic argument behind \eqref{eq:conj} and a replica computation from \cite{CrisantiLeuzzi}. It is further explained in Section \ref{sec:heuristic}. As mentioned above the condition $A(q,\beta) \le 0$ is stronger than Plefka's condition. However, at least for pure $p$-spin models it is only slightly stronger, in that it ultimately only additionally determines when the local maximum $m=0$ should be considered relevant; with high probability there are no other local maxima where $A(q,\beta) \le 0$ is satisfied and Plefka's condition is not, see Theorem \ref{thm:condequiv} for more details. The stronger condition \emph{is} necessary to get a coherent phase diagram. For instance if only Plefka's condition is required for a TAP local maximum to be $m$ considered relevant, the claim \eqref{eq:conj} can not be true, since $m=0$ always satisfies Plefka's condition and would always be a relevant TAP solution and we would have $f_{\TAP}(\beta)\ge\frac{\beta^2}{2}$ for all $\beta \ge 0$.

\subsection{Relation to \cite{subagFreeEnergySpherical2021,arousShatteringMetastabilitySpin2021}}

While preparing the current article the two related works \cite{subagFreeEnergySpherical2021,arousShatteringMetastabilitySpin2021} appeared.

The article \cite{subagFreeEnergySpherical2021} computes the limit $N\to\infty$ of the free energy \eqref{eq: FE} of the spherical models considered in the paper for all $\beta$ using a TAP approach, though one involving limiting properties of the Gibbs measure via the concept of ``multisamplable overlap'' which is therefore different from TAP approach envisioned in the heuristic in Section \ref{sec:heuristic}. In the framework of \cite{subagFreeEnergySpherical2021} only the ground state energy $E_0$ of $H_N$ plays a role, rather than the full TAP complexity. Therefore the phase transitions $\beta_c$ and $\beta_d$ can not be detected in that framework. The phase transition $\beta_s$ can however be detected, and \cite{subagFreeEnergySpherical2021} presents a formula for it which is different but can be shown to be equivalent to the the formula \eqref{betasdef} in the paper (see \cite[(1.9)-(1.10)]{subagFreeEnergySpherical2021}). It also presents the formula $h_{\rm{TAP}}(V_*,q_*)$ which we also derive in this paper (see \eqref{eq: fTAP static low temp} and \cite[(1.11)-(12)]{subagFreeEnergySpherical2021}). In contrast to the present paper, which only deals with the variational principle \eqref{eq:varprinc}, the paper \cite{subagFreeEnergySpherical2021} computes the limiting free energy. It proves that for $\beta\le \beta_s$ it holds that $\lim_{N\to\infty}f_N(\beta) = \frac{\beta^2}{2}$ and for $\beta \ge \beta_s$ one has $\lim_{N\to\infty}f_N(\beta) = h_{\rm{TAP}}(V_*,q_*)$, where $V_*,q_*$ are as in Theorem \ref{thm:fstthrm}. Since Theorems \ref{thm:comp}-\ref{thm:dynLTstatHT} show that $\fTAP(\beta)=\frac{\beta^2}{2}$ for $\beta \le \beta_s$ and Theorem \ref{thm:fstthrm} shows that $\fTAP(\beta)=h_{\rm{TAP}}(V_*,q_*)$ we can ``aposteriori'' conclude that \eqref{eq:conj} is indeed true. However, in the TAP approach envisioned by \cite{DavidNicola,TAPHTUB} and the present work one wishes to prove this rather by obtaining a direct microcanonical proof of \eqref{eq:conj}, which would then yield an alternative proof of the results for the limiting free energy of \cite{subagFreeEnergyLandscapes2020} when combined with the present paper.

The ``TAP decomposition'' of \cite{arousShatteringMetastabilitySpin2021} is more similar to the TAP approach envisioned here, and here the analysis is sensitive to the threshold $\beta_d$. Indeed \cite[Theorems 2.1, 2.4]{arousShatteringMetastabilitySpin2021} proves that there is a $\delta>0$ such that for $\beta \in (\beta_d-\delta,\beta_d)$ the free energy can be lower bounded by the contribution of exponentially many ``slices'' around TAP solutions, giving a total contribution of $\beta^2/2$ (cf. Section \ref{sec:heuristic} and Theorem \ref{thm:dynLTstatHT}). A similar computation of the free energy for large enough $\beta$ was carried out in \cite{subagGeometryGibbsMeasure2017}. This is the kind of analysis that the authors hope to in the future extend to all $\beta$, whereby the aim is to separate the analysis neatly into a proof of \eqref{eq:conj} for all $\beta$ (a first step has been taken in \cite{TAPHTUB}) and an analysis of variational principle \eqref{eq:varprinc} for all $\beta$, which is provided by the present paper.

\subsection{Further results and structure of paper}
Sections \ref{sec:det}, \ref{computeI} and \ref{sec: umax umin} contain further results that are of independent interest beyond their role as intermediate steps in the proofs of Theorem \ref{thm:comp}-\ref{thm:fstthrm}. Theorem \ref{thm: det char f} of Section \ref{sec:det} gives various results that deterministically relate TAP energy $H_{\TAP}(m)$, Hamiltonian energy $H_N\left(\frac{m}{|m|}\right)$ and magnitude $|m|^2$ for any \TP\ solution, and that follow purely from the conditions $|m|^2 \in D_\beta$ and that a \TP\ solution must be a local maximum.
Theorem \ref{thm:TAP comp} of Section \ref{computeI} gives a more detailed version of Theorem \ref{thm: TAP comp intro}. Proposition \ref{prop: umax umin formulas} gives formulas for $U_{\min}$ and $U_{\max}$. Theorems \ref{thm:sndthrm}-\ref{thm:fstthrm} are proved in Section \ref{optimize}.
In Section \ref{sec:heuristic} we give a heuristic behind \eqref{eq:conj} and the condition \eqref{eq: our plefka}. Section \ref{sec: prel} recalls some known facts about $H_N$, including the full definition of the complexity $I$ and $E_0$, that we use in this paper.

Table \ref{tab:notation} contains a list of notation used in this paper.

\begin{table}[H]
  \centering
    \def\arraystretch{1.9}
    \begin{tabular}{| p{0.15\textwidth} | p{0.6\textwidth} | p{0.20\textwidth}|}
    \hline
    Notation & Description & Definition\\
    \hline
    \hline
    $H_N(m)$ & Hamiltonian energy & \eqref{eq: Hamiltonian covar}\\
    
    $H_{\TAP}(m)$ & TAP energy as function of $m$ & \eqref{eq:HTAPdef}\\
    
    $h_{\TAP}(E,|m|^2)$ & TAP energy as function of Hamiltonian energy $E=H_N(m)$ and $|m|^2$ & \eqref{eq: hTAP def}\\
    
    $E_0, E_\infty$ & Largest and smallest energies of $p$-spin Hamiltonian local maxima & \eqref{eq: Einf def} and \eqref{eq: E0 zero of I}\\
    $\beta_2(q)$ & Quantity appearing in Plefka's condition& \eqref{eq: beta2 def}\\
    $\xi(x)=x^p$ & Covariance function of $p$-spin Hamiltonian & \eqref{eq: xi def}\\ 
    $\On(q)$ & Onsager correction & \eqref{eq: Onsager Def}\\  

    $f(E,q)$  & $h_{\TAP}(q^{p/2}E,q)$ & \eqref{eq: f def}\\
    $D_\beta, A(q,\beta)$ & $D_\beta=\{q:A(q,\beta)\leq 0\}$ set of possible squared radii of relevant TAP solutions & \eqref{eq: our plefka}, \eqref{eq: Adef}\\
    $\ETAPmax, \ETAPmin$
    & Largest and smallest TAP energies of nonzero \TP\  solutions & \eqref{eq: Uminbeta}, \eqref{eq: Umax Umin def}\\
    $V_{*}, \ETAPstar, \qstar $ & Energy, $\TAP$ energy and squared radius with largest contribution to \TP\ free energy & \eqref{eq: def E TAP star}\\
    $\Eq, \qE$ & Energy of nonzero \TP\  solution of given squared radius, and squared radius of nonzero \TP\ solution of given energy & \eqref{eq: Eq def}, \eqref{eq: qE def}\\
    $\EU,\qU$ & Energy, squared radius of nonzero \TP\  solution for given $\TAP$ energy & \eqref{eq: final def EU}, \eqref{eq: qU def}\\
    $\ITAP(U,V,q)$ & Entropy of \TP\ local maxima of TAP energy & \eqref{eq: ITAP extended}\\
    $E_{\min },q_{\min}$ & Minimal energy on unit sphere of nonzero \TP\ solutions and corresponding squared radius& \eqref{eq: Emin def}, \eqref{eq: qmin def}\\
    $r_{\pm}\left(\frac{E}{E_\infty}\right)$  & Solutions of $\frac{E}{E_{\infty} }= \frac{1}{2}(\frac{1}{x}+x)$ with $r_{-}\left(\frac{E}{E_\infty}\right)\le1\le r_{+}\left(\frac{E}{E_\infty}\right)$ & \eqref{eq: rpmdef} \\
    $\bar{r}$  & $r_{-}(\frac{E_0}{E_\infty})$  & \eqref{eq: rbardef} \\
    $\NU$  & Number of \TP\ solutions with TAP energy, Hamiltonian energy and squared radius in given sets  & \eqref{Ndef} \\
    \hline    
    \end{tabular}        
  \caption{Index of notation\label{tab:notation}}
\end{table}
\section{Heuristic derivation of the \TP\ variational principle}\label{sec:heuristic}

In this section we give a heuristic derivation of the TAP energy \eqref{eq:HTAPdef} and \eqref{eq:conj}. It is an adaptation of the heuristic that has been turned into a proof of \eqref{eq:conj} in the special case $p=2$ in \cite{DavidNicola}, and an upper bound for the free energy in terms of the TAP energy in \cite{TAPHTUB}. The heuristic also motivates the condition \eqref{eq: our plefka}.

The starting point is that in high temperature and without external
field the free energy of a Hamiltonian $H_N$ whose covariance is given by $\xi$ takes a simple form
\begin{equation}\label{eq: paramag FE}
Z_{N}=e^{N\frac{\beta^{2}}{2}\xi(1)+o\left(N\right)}.
\end{equation}
The estimate \eqref{eq: paramag FE} of course does not hold in low temperature. In this heuristic
we make the ansatz that \eqref{eq: paramag FE} is true at least in the region reported
in \cite{CrisantiLeuzzi} as featuring stability in the replica computation. Their condition can be written as
\begin{equation}\label{eq: internal stability}
\sup_{r \in [0,1]} \left( \beta^2 \frac{\xi'(q)}{r} - \frac{1}{1-r} \right) \ge 0.
\end{equation}
To argue heuristically that in low temperature the free
energy can be written in terms of the TAP energy, we first lower
bound the partition function by the integral of the Gibbs factor
over a ``slice'' $A_\varepsilon =\left\{ \sigma\in S:\left|\left(\sigma-m\right)\cdot m\right|\le\varepsilon\right\} $
for some $m\in B$:
\begin{equation}\label{eq:slidelb}
Z_{N}\ge\int_{A_\varepsilon}e^{\beta H_{N}\left(\sigma\right)}E(d\sigma).
\end{equation}
The set $A_\varepsilon$ is an $\varepsilon$-thickened version of the intersection $A_0$
of the hyperplane perpendicular to $m$ passing through $m$, and
the sphere. The set $A_0$ is precisely a $N-2$-dimensional sphere of radius $\sqrt{1-|m|^2}$ which has surface area $(1-|m|)^{\frac{N}{2}+o(N)}$. For 
 $\varepsilon \downarrow 0$ slow enough with $N$. The measure of $A_\varepsilon$ under $E$ is also $(1-|m|)^{\frac{N}{2}+o(N)}$.

After normalization the integral in \eqref{eq:slidelb} it can be approximated by an uniform integral on $A_0$, giving that
\begin{equation}\label{eq: lb after norm}
Z_{N} \ge \exp\left( \frac{N}{2}\log(1-|m|^2) + o(N) \right) \int_{A_0}e^{\beta H_{N}\left(\sigma\right)}d\sigma,
\end{equation}
where the integral is now against the uniform measure on $A_0$.

Inside the slice $A_0$, it is natural to expand the Hamiltonian 
in $\hat{\sigma}=\sigma-m$ giving
\begin{equation}\label{eq: expansion}
H_{N}\left(m+\hat{\sigma}\right)=H_{N}\left(m\right)+\nabla H_{N}\left(m\right)\cdot\hat{\sigma}+H_{N}^{m}\left(\hat{\sigma}\right),
\end{equation}
where $H_{N}^{m}\left(\hat{\sigma}\right)$ collects all the terms
of order $2$ or higher in the $\hat{\sigma}_{i}$. One can show
that for fixed $m$
\[
H_{N}\left(m\right),\left(\nabla H_{N}\left(m\right)\cdot\hat{\sigma}\right)_{\hat{\sigma}:\hat{\sigma}\cdot m=0},\left(H_{N}^{m}\left(\hat{\sigma}\right)\right)_{\hat{\sigma}:\hat{\sigma}\cdot m=0},
\]
are independent, and that $H_{N}^{m}\left(\hat{\sigma}\right)$ is
a mixed $p$-spin Hamiltonian on the $N-2$-dimensional sphere $\left\{ \hat{\sigma}:\hat{\sigma}\cdot m=0\right\} $
with covariance
\[
\mathbb{E}\left[H_{N}^{m}\left(\hat{\sigma}\right)H_{N}^{m}\left(\hat{\sigma}'\right)\right]=\xi^{|m|^2}( \hat{\sigma} \cdot \hat{\sigma}'),
\]
where
\[
\xi^{q}\left(x\right)=\xi\left(q+x\left(1-q\right)\right)-\xi'\left(q\right)x\left(1-q\right)-\xi\left(q\right),
\]
(see Lemma 3.2 \cite{TAPHTUB}). With \eqref{eq: expansion} the integral in \eqref{eq: lb after norm} can be written as
\begin{equation}\label{eq: LB slice}
    e^{\beta H_N(m)} \int_{A_0}e^{\beta\nabla H_{N}\left(m\right)\cdot\hat{\sigma}+\beta H_{N}\left(\hat{\sigma}\right)}d\sigma.
\end{equation}
This integral reveals itself as the partition function of a spherical
model on $A_0$ with external field $\beta \nabla H_{N}\left(m\right)$ and Hamiltonian
$H_{N}^{m}\left(\hat{\sigma}\right)$. Rewriting in this way is useful
if we can expect this integral to be estimated in a simpler way than the original one, by a simpler expression. If the external field field vanishes, and if $\beta$ is small enough depending on the covariance
$\xi^{m}$ then the expression \eqref{eq: paramag FE} gives such a simple expression. We therefore
restrict our attention only to $m$:s such that $\nabla H_{N}\left(m\right)\propto m$
so that the partition function on the bottom line of \eqref{eq: LB slice} has no external field, and the covariance $\xi^{|m|^2}$ of the slice satisfies \eqref{eq: internal stability}. The  covariance $\xi^{|m|^2}$ satisfies \eqref{eq: internal stability} precisely if \eqref{eq: our plefka} holds; this is the motivation for \eqref{eq: our plefka}.
When this is the case we can heuristically use \eqref{eq: paramag FE} on the partition function in \eqref{eq: LB slice} and we obtain that \eqref{eq: LB slice} should equal
$$e^{\beta H_N(m) + \frac{\beta^2}{2}\xi_q(1)},$$
for any fixed $m$ such that \eqref{eq: our plefka} and $\nabla H_N(m) \propto m$ are satisfied. For such $m$ , it would follow from \eqref{eq: LB slice} that
\begin{equation}
Z_{N} \ge e^{\beta H_N(m) + \frac{N}{2}\log(1-|m|^2) + \frac{\beta^2}{2}\xi_q(1)} = e^{H_{\TAP}(m) + o(N)}.
\end{equation}
Thus we have heuristically arrived at the formula \eqref{eq:HTAPdef} for the TAP energy.

To obtain the best possible lower bound, it is natural to maximize $H_{\TAP}(m)$, leading one to consider $m$ that are maximizers of $H_{\TAP}(m)$. These will be critical points of $H_{\TAP}(m)$, which because of the spherical symmetry of all terms in $H_{\TAP}$ except $H_N(m)$ means that $m$ will be a critical point of $H_N$ in the spherical metric, which incidentally is equivalent to the condition $\nabla {H_N}(m) \propto m$ assumed above to find a heuristic lower bound for the partition function.
Thus, heuristically, we arrive at the lower bound
\begin{equation}
\begin{array}{ccl}
Z_{N} \ge e^{H_{\TAP}(m) + o(N)}\text{ for any local maximum $m$ of $H_{\TAP}$ satisfying \eqref{eq: our plefka}}.
\end{array}
\end{equation}

If there are many local maxima, it is natural that these need to be added up to give the true magnitude of $Z_N$. Assuming that any over-counting arising in this way causes only lower order errors, we heuristically arrive at the estimate
$$Z_N = \sum_{m:\text{ loc max of } H_{\TAP}, |m|^2 \in D_\beta} e^{H_{\TAP}(m) + o(N)}.$$
Since
$$ \lim_{N \to\infty} \frac{1}{N}\log \left( \sum_{m:\text{ loc max of } H_{\TAP}, |m|^2 \in D_\beta} e^{H_{\TAP}(m) + o(N)} \right) = \fTAP(\beta),$$
we arrive heuristically at \eqref{eq:conj}.

One can show that \eqref{eq: internal stability} is equivalent to $\beta < \beta_d$ when $\xi(x)=x^p$, so that if \eqref{eq:conj} is true, then
it follows from Theorem \ref{thm:dynLTstatHT} that the estimate \eqref{eq: paramag FE} in fact remains true
also for a range of $\beta$ that do not satisfy \eqref{eq: internal stability}, but for a very different reason, as comparing Theorems \ref{thm:sndthrm} and \ref{thm:dynLTstatHT} shows.

Proving that \eqref{eq: internal stability} indeed implies \eqref{eq: paramag FE} is the subject of active research and beyond
the scope of the present article. A direct rigorous proof of such an implication will likely
combine with the theorems of this paper to yield a fully rigorous
computation of the free energy via a TAP approach. The condition \eqref{eq: internal stability} arises from a replica calculation. It is possible that a different method would lead to a different, but ultimately equivalent condition. We hope that future work will find a direct proof of \eqref{eq: internal stability} that gives rise to a such a condition. It is also possible that a non-equivalent condition is obtained. However, in the proofs of the present article we only use the following properties of the set $D_\beta$.
\begin{lem}\label{lem:ourplefkaprofs}
The family of sets $(D_\beta)_{\beta \ge 0}$ in \eqref{eq: our plefka} satisfies
 \begin{enumerate}
 \item 
  $D_\beta \cap \left[ \frac{p-2}{p},1\right] = \left\{q \in \left[ \frac{p-2}{p},1\right]: \beta_2(q) \le \frac{1}{\sqrt{2}}\right\}.$
 \item $D_\beta \cap \left[0, \frac{p-2}{p}\right)\subset \left\{q \in \left[ 0,\frac{p-2}{p}\right): \beta_2(q) \le \frac{1}{\sqrt{2}}\right\}$.
 \item $0\in D_\beta \iff \beta \leq \beta_d$. \end{enumerate}
\end{lem}
As these are the only properties of $D_\beta$ used, all our results will remain true if our condition is replaced by any other condition also satisfying these properties:
\begin{teor}\label{thm:condequiv}
 If $(D_\beta)_{\beta \ge 0}$ in \eqref{eq: our plefka} is replaced by any collection of sets indexed by $\beta$ that satisfy Lemma \ref{lem:ourplefkaprofs} 1), 2) and 3) then all the results stated in the introduction remain true. In fact everything except the proof of Lemma 2.1 remains exactly the same. 
\end{teor}
\subsection{Relation between conditions}\label{sec:CondRel}
In this section we prove the properties of $D_\beta$ stated in Lemma \ref{lem:ourplefkaprofs}, that are needed for the analysis in this paper. 
Before stating the results we recall \eqref{eq: our plefka} and \eqref{eq: Adef} which state that $D_\beta= \{q\in[0,1]: A(q,\beta)\leq 0\}$, where
\begin{equation}\label{eq: Adef 2}
    A\left(q,\beta\right)=\sup_{r\in\left(0,1\right)}\left( \beta^{2}\frac{\xi'\left(q+r\left(1-q\right)\right)(1-q)-\xi'\left(q\right)\left(1-q\right)}{r}-\frac{1}{1-r}\right).
\end{equation}
\begin{proof}[Prof of Lemma \ref{lem:ourplefkaprofs} 1) 2)]
We first show that for $q\in[0,1]$, $\beta\geq0$ we have
\begin{equation}\label{eq: condrel1}
    A(q,\beta) \leq 0 \Rightarrow \beta_2(q)\leq \frac{1}{\sqrt{2}}.
\end{equation}     
Let $A(q,\beta)\leq 0$. Then by the definition of \eqref{eq: Adef 2} of $A$
\[ 0 \geq \lim_{r\searrow 0} \left( \beta^2(1-q) \frac{\xi'(q+r(1-q))-\xi'(q)}{r}-\frac{1}{1-r} \right) = \beta^2 (1-q)^2 \xi''(q) - 1 \overset{\eqref{eq: beta2 def}}{=} 2\beta_2(q)^2-1,\]
so \eqref{eq: condrel1} follows. This proves Lemma \ref{lem:ourplefkaprofs} 2).

Next we show for $q\geq \frac{p-2}{p}$, $\beta \geq 0$, that we have 
\begin{equation}\label{eq: condrel2}
    \beta_2(q)\leq \frac{1}{\sqrt{2}} \Rightarrow  A(q,\beta) \leq 0,
\end{equation}
which together with \eqref{eq: condrel1} also implies Lemma \ref{lem:ourplefkaprofs} 1).
Since $\xi'\left(x\right)=px^{p-1}$ we have using the definition
\eqref{eq: beta2 def} of $\beta_{2}$ that
\begin{equation}\label{eq: first ineq}
\beta^{2}\left(1-q\right)\frac{\xi'\left(q+r\left(1-q\right)\right)-\xi'\left(q\right)}{r}=2\beta_{2}\left(q\right)^{2}\frac{\left(1+r\frac{1-q}{q}\right)^{p-1}-1}{\left(p-1\right)r\frac{1-q}{q}}.
\end{equation}
Consider
\begin{equation}
\frac{\left(1+r\frac{1-q}{q}\right)^{p-1}-1}{\left(p-1\right)r\frac{1-q}{q}}.
\end{equation}
By the binomial theorem the left-hand equals
\[
\frac{1}{p-1}\sum_{k=0}^{p-2}{{p-1}\choose{k+1}}\left(r\frac{1-q}{q}\right)^{k}. 
\]
Now if $q\ge\frac{p-2}{p}$ then $\frac{1-q}{q}\le\frac{2}{p-2}$
so that this is at most
\[
\frac{1}{p-1}\sum_{k=0}^{p-2}{p-1 \choose k+1}\left(\frac{2}{p-2}\right)^{k}r^{k}.
\]
Using the inequalities ${p-1 \choose k+1}\le\frac{\left(p-1\right)\left(p-2\right)^{k}}{\left(k+1\right)!}$
and $\frac{2^{k}}{\left(k+1\right)!}\le1$ we get that this is at
most $\sum_{k=0}^{p-2}r^{k}\le\frac{1}{1-r}$. We thus have for all $q\ge\frac{p-2}{p}$ and $r\in\left[0,1\right]$
that
\begin{equation}
\frac{\left(1+r\frac{1-q}{q}\right)^{p-1}-1}{\left(p-1\right)r\frac{1-q}{q}}\le\frac{1}{1-r}.\label{eq: elem ineq}
\end{equation}
Combining this with \eqref{eq: first ineq} we obtain that $A\left(q,\beta\right)\le0$
when $2\beta_{2}\left(q\right)^{2}\le1$.

\end{proof}
\begin{proof}[Proof of Lemma \ref{lem:ourplefkaprofs} 3)]
We show that
\[ A(0,\beta) \leq 0\Leftrightarrow \beta\leq \beta_d.\]
By definition \eqref{eq: Adef 2} of $A$  the expression $A(0,\beta)\leq 0$ reads 
\[\sup_{r\in\left(0,1\right)}\left( \beta^{2}\frac{\xi'\left(r\right)}{r}-\frac{1}{1-r}\right)\leq 0. \]
Since $\xi'(x)=p x^{p-1}$ this is equivalent to 
\[\sup_{r\in(0,1)} \left( r^{p-2}(1-r) \right) \leq \frac{1}{p\beta^{2}}.  \]
The left hand side is easily checked to be maximized at $r=\frac{p-2}{p-1}$. Thus $A(0,\beta)\le0$ is equivalent to
\[ \frac{(p-2)^{p-2}}{(p-1)^{p-1}} \leq \frac{1}{p\beta^{2}},  \]
which is equivalent to $\beta \le \beta_d$ by the definition \eqref{eq: beta d def} of $\beta_d$.

\end{proof}

\section{Preliminaries}\label{sec: prel}
\subsection{Hamiltonian as random homogeneous polynomial}
We record the standard fact that the Hamiltonian in \eqref{eq: Hamiltonian covar} with $\xi(x)=x^p,p \ge 1,$ can be explicitly constructed by letting 
\begin{equation}\label{eq:hamiltonianassum}
H_{N}(\sigma) = \sqrt{N} \sum_{i_1,\ldots,i_p = 1}^{N} J_{i_1,\ldots,i_p} \sigma_{i_1}\ldots\sigma_{i_p}, \sigma \in \R^N,
\end{equation}
where $J_{i_1,\ldots,i_p}$ are independent standard Gaussians. This implies that $H_N(\sigma)$ is almost surely $p$-homogenous, which will be crucial.
\subsection{Critical point complexity of Hamiltonian}\label{sec: crit point complexity}
In this subsection we recall the precise form of the critical point complexity for pure $p$-spin models.
Let
$$\mathcal{C}_N( A ) = \frac{1}{N} \log | \{ \sigma \in S_{N-1}: \nabla_{\rm{sp}} H_N(\sigma)=0, \frac{1}{N} H_N(\sigma)\in A \} |,$$
where $\nabla_{\rm{sp}}$ denotes the spherical gradient, be the number of critical points of $H_N$ in the spherical metric with scaled energy in the set $A$, and let
$$\mathcal{M}_N( A ) = \frac{1}{N} \log | \{ \sigma \in S_{N-1}: \nabla_{\rm{sp}} H_N(\sigma)=0, m \text{ is loc. max.}, \frac{1}{N} H_N(\sigma)\in A \} |,$$
be the same for local maxima.
Let the log potential of the semi-circle law  $\mu_{sc}$ (with support on $[-1,1]$) be denoted by

\begin{equation}\label{eq: omega}
    \Omega(\eta) = \int \log | \eta - x | \mu_{sc}(dx) = \eta^2 - \frac{1}{2} - \eta \sqrt{\eta^2-1} + \log( \eta + \sqrt{\eta^2-1} ),
\end{equation}
for $\eta \ge 1$. Also let
\begin{equation}\label{eq:g}
g\left(\eta\right)=\begin{cases}
-\infty & \text{ if }\eta<1,\\
\frac{1}{2} + \frac{1}{2}\log(p-1) - 2 \frac{p-1}{p} \eta^2 + \Omega(\eta) & \text{ if }\eta\ge1.
\end{cases}
\end{equation}
Then the annealed complexity is given by the function
\begin{equation}\label{eq:Iann def intro}
I_{\rm{Ann}}\left(E\right)=g\left(\frac{E}{E_\infty}\right),
\end{equation}
and the quenched by
\begin{equation}\label{Thetadef}
I\left(E\right)=\begin{cases}
g\left(\frac{E}{E_\infty}\right) & \text{\,if }g\left(\frac{E}{E_\infty}\right)\ge0,\\
-\infty & \text{if }g\left(\frac{E}{E_\infty}\right)<0,
\end{cases}
\end{equation}
One can verify that
\begin{equation}\label{eq: I decreasing}
    I_{\Ann}(E)\text{ is strictly decreasing on }[E_\infty,\infty),
\end{equation}
and using the notation of \cite{purefstmom} one
\begin{equation}\label{eq: E0 zero of I}
\text{denotes by } E_0 \text{ the unique zero of } I_{\Ann} \text{ in } [E_\infty,\infty),
\end{equation}
so that
\begin{equation}\label{eq: I props}
I\left(E\right)=\begin{cases}
g\left(\frac{E}{E_{\infty}}\right)>0 & \text{ for }E\in[E_{\infty},E_{0}),\\
0 & \text{ for }E = E_{0},\\
-\infty & \text{\,for }E\notin\left[E_{\infty},E_{0}\right],
\end{cases}
\end{equation}
We have the following.
\begin{teor}[\cite{purefstmom}, \cite{subagpure}]\label{thm: comp crit point}
For all $E$
$$ \lim_{\varepsilon \downarrow 0} \lim_{N\to\infty} \mathcal{C}_N( [E-\delta,E+\delta]) = I(E),$$
where the convergence is in probability.
\end{teor}
Note that Theorem \ref{thm: comp crit point} and \eqref{eq: I props} imply \eqref{eq: ground state}.

From Theorem \ref{thm: comp crit point} one easily derives the equivalent result for local maxima.
\begin{cor}
For all $E$
$$ \lim_{\varepsilon \downarrow 0} \lim_{N\to\infty} \mathcal{M}_N( [E-\delta,E+\delta]) = I(E),$$
where the convergence is in probability.
\end{cor}
\begin{proof}
Let $\mathcal{C}([E,\infty))$ denote the number of all critical
points of $H_{N}\left(\sigma\right)$ on $S_{N-1}$ with $\frac{H_{N}\left(\sigma\right)}{N}\in[E,\infty)$,
and let $\mathcal{M}\left(\left[E,\infty\right)\right)\le\mathcal{C}\left(\left[E,\infty\right)\right)$ denote
the number of local maxima satisfying the same condition. Theorems 2.5 and
2.8 \cite{purefstmom} show that
\[
\lim_{N\to\infty}\frac{1}{N}\log\mathbb{E}\left[\mathcal{C}([E,\infty))\right]=\lim_{N\to\infty}\frac{1}{N}\log\mathbb{E}\left[\mathcal{M}([E,\infty))\right] = g\left(\frac{E}{E_{\infty}}\right)\text{ for all }E\ge E_{\infty},
\]
for $g$ given by \eqref{eq:g} (cf. (2.15)-(2.16) of \cite{purefstmom}; the results of \cite{purefstmom} and \cite{subagpure} are stated for negative energies and local minima, since $H_N \overset{law}{=} -H_N$ the equivalent results for local maxima stated here and below follow). Since $\mathcal{M}([E,\infty))$ is an integer it follows by Markov's inequality that
\[
\lim_{N\to\infty}\frac{1}{N}\log\mathcal{C}([E,\infty))=\lim_{N\to\infty}\frac{1}{N}\log\mathcal{M}([E,\infty))=-\infty\text{\,for }E\ge E_{0},
\]
for $E_{0}$ defined below \eqref{eq: loc max quenched}, where the limits are in probability. Corollary
2 of \cite{subagpure} implies that 
\begin{equation}
\lim_{N\to\infty}\frac{1}{N}\log\mathcal{C}([E,\infty))=I\left(E\right)\text{\,for }E\in\left(E_{\infty},E_{0}\right).\label{eq: subag LB}
\end{equation}
Theorem 2.5 \cite{purefstmom} shows that for any fixed $k \ge 1$ the number of critical points of index $k$ is much smaller than $e^{N I(E)}$, strongly suggesting that \eqref{eq:local max interval} below follows. To also cover the case of diverging $k$ we invoke Theorem 2.15 \cite{purefstmom} and the fact that $E_k$ ($E_k(3)$ in the notation of \cite{purefstmom}) satisfies $\lim_{k\to\infty} E_k = E_\infty$. The latter shows that for any $E>E_\infty$ there is a $K$ such that $E>E_K$. Then using Theorem 2.5 \cite{purefstmom} for critical points of index $1,\ldots,K$ and Theorem 2.15 \cite{purefstmom} for indices larger than $K$ we get
\begin{equation}\label{eq: ub other index}
\lim_{N\to\infty}\frac{1}{N}\log\left(\mathcal{C}([E,\infty))-\mathcal{M}([E,\infty))\right)\le I(E) - I_1(E),
\end{equation}
for $I_{1}$ as in (2.14) of \cite{purefstmom}, which is positive for all $E\in\left(E_{\infty},E_{0}\right)$.
From (\ref{eq: subag LB}) and (\ref{eq: ub other index}) it follows
that in fact
\begin{equation}\label{eq:local max interval}
    \lim_{N\to\infty}\frac{1}{N}\log\mathcal{M}([E,\infty))=I\left(E\right)\text{\,for }E\in\left(E_{\infty},E_{0}\right).
\end{equation}
Since $I\left(E\right)$ is strictly decreasing for $E\in\left(E_{\infty},E_{0}\right)$ the claim \eqref{eq: loc max quenched} follows.
\end{proof}

\section{Deterministic characterization of \TP\ solutions}\label{sec:det}
In this section we derive a characterization of \TP\ solutions that arises determinstically from the condition that $m$ must be a local maximum of $H_{\TAP}$ and satisfy $|m|^2 \in D_\beta$, together with two basic deterministic properties of $H_N$ (namely \eqref{eq: g assumptions} below). We will make no reference to the random behavior of $H_N$. 

To formulate the results  define for
\begin{equation}\label{eq: g assumptions}
\text{any $p$-homogeneous twice differentiable function $g:B_N\to\mathbb{R}$
 for $p\ge3$}
\end{equation}
the $g$-TAP energy 
\begin{equation}\label{eq:HTAPgdef}
\begin{array}{rcl}
\HTAPg\left(m\right) &= &
N h_{\TAP}\left( \HNg(m), |m|^2\right)\\
&\overset{\eqref{eq: hTAP def}}{=} & N \beta  g\left(m\right)+\frac{N}{2}\log\left(1-\left|m\right|^{2}\right)
+
N\On(|m|^2),
\end{array}
\end{equation}
and say that $m$ is a \gTAP\, solution if $\nabla H_{\TAP}^{g}\left(m\right)=0$.
If $m$ is a local maximum of $H_{\TAP}^{g}\left(m\right)$ and $|m|^2 \in D_\beta$ we call
it a relevant \gTAP\, solution. 

The energy $\frac{1}{N}H_N(m)$ almost surely satisfies the conditions of \eqref{eq: g assumptions} (as can be seen from \eqref{eq:hamiltonianassum}), and the subsequent sections will use the results of this section with $g(m)=\frac{1}{N}H_N(m)$. With this choice a (relevant) \gTAP\, solution is a (relevant) TAP solution, and $\HTAPg=H_{\TAP}$.

There is a mapping between \gTAP\, solutions and local maxima of $g$. To see this, note that all terms in the bottom line of \eqref{eq:HTAPgdef} except the term $N \beta g(m)$ are spherically symmetric, so that any non-zero local maximum $m$ of $\HTAPg(m)$ must also be a local maximum of $\HNg$ on any sphere $\{ \sigma: |\sigma|^2=q\}$, $q\in (0,1]$.
Using also that $\HNg$ is $p$-homogenous and letting $\hat{m}$ denote $m/|m|$, we have that $\mhattext$ is a local maximum of $\HNg$ on $S_{N-1}$. Conversely, if $\mhattext$ is a local maximum of $\HNg$ on $S_{N-1}$ then $m$ is local maximum of $\HTAPg$ if it is also a local maximum in the radial direction, that is if
$$q \to h_{\TAP}\left( q^{p/2} \HNg(\hat{m}), q\right)$$
has a local maximum at $q=|m|^2$. For brevity let
\begin{equation}\label{eq: f def}
\begin{array}{rcl}
    f(E,q) & = & h_{\TAP}(q^{p/2}E,q)\\
           & \overset{\eqref{eq: hTAP def},\eqref{eq: Onsager Def}}{=} & \beta q^{p/2}E + \frac{1}{2}\log(1-q) + \frac{\beta^{2}}{2} (\xi(1) - \xi'(q)(1-q) - \xi(q)),
\end{array}
\end{equation}
so that for all $m$
\begin{equation}\label{eq: HTAP to f}
    \HTAPg(m) = f( \HNg(\hat{m}), |m|^2).
\end{equation}
We then have:
\begin{lem}\label{lem: det TAP sol radial sym}
For any $g$ as in \eqref{eq: g assumptions}:
\begin{enumerate}
\item $m\in B_N\setminus\{0\}$ is a relevant \gTAP\, solution iff $\left|m\right|^2 \in D_{\beta}$,
$\mhatdisplay$ is a local maximum of $g$ on $S_{N-1}$ and $|m|^2$ is a local maximum of $q\to f\left(\HNg(\hat{m}),q\right)$.

\item $m=0$ is always a local maximum of $\HTAPg\left(m\right)$ and iff $\beta\le\beta_{d}$ it is a relevant \gTAP\ solution.
\end{enumerate}
\end{lem}
\begin{proof}
$ $
\begin{enumerate}
\item This follows from the considerations in the paragraph before the lemma.\\
\item
The entropy term $\frac{1}{2}\log(1-|m|^2)$ of \eqref{eq: hTAP def} has zero gradient and negative definite Hessian at $m=0$. By \eqref{eq: g assumptions} the term $g$ has both vanishing gradient and vanishing Hessian at $m=0$. Also since
\begin{equation}\label{eq: Onsager concrete}
\On(q) = \frac{\beta^{2}}{2} (\xi(1) - \xi'(q)(1-q) - \xi(q)) \overset{\xi(x)=x^p}{=}\frac{\beta^{2}}{2}\left(1-p\left(1-q\right)q^{p-1}-q^{p}\right),
\end{equation}
and $p\ge3$ so does the term $\On(|m|^2)$. Therefore $m=0$ is always a local maximum of $\HTAPg(m)$. Thus $m=0$ is a relevant \gTAP\ solution iff  $m \in D_\beta$, which by Lemma \ref{lem:ourplefkaprofs} 3) is equivalent to $\beta \le \beta_d$.
\end{enumerate}
\end{proof}

Next we will give a complete analysis of the critical points of $q\to f(E,q)$ for different values of $\beta$ and $E$, thus determining for each $\beta$ and $E$ which values of $|m|^2=q$ (if any) are possible for a relevant \gTAP\ solution arising from a critical point $\mhatdisplay$ with $\HNg(\hat{m})=E$.
This rests on the next lemma. Before we state it, let
\begin{equation}\label{eq: rpmdef}
    r_{\pm}(x) = x \pm \sqrt{x^2-1},
\end{equation}
and note that
\begin{equation}\label{eq: rpmdef as sols}
\begin{array}{l}
r_{+}:[1,\infty)\to[1,\infty) \text{ is the inverse of } z \to \frac{1}{2}(\frac{1}{z}+z)\text{ for }  z\ge1 \text{ and is increasing,}\hspace{10pt}\\
r_{-}:[1,\infty)\to(0,1] \text{ is the inverse of } z \to \frac{1}{2}(\frac{1}{z}+z)\text{ for }  0< z\le1 \text{ and is decreasing.}\\\end{array}
\end{equation}
\begin{lem}\label{lem: critequiv}
It holds that
\begin{equation}\label{eq: f deriv}
    \partial_q f(E,q) = \frac{\sqrt{2}\beta_2(q)}{1-q} \left\{ \frac{E}{E_\infty} - \frac{1}{2}\left(\frac{1}{\sqrt{2}\beta_{2}\left(q\right)}+\sqrt{2}\beta_{2}\left(q\right)\right)\right\}\text{ for all }E,q.
\end{equation}     
The critical point equation 
\begin{equation}\label{eq: crit point eq}
    \partial_q f(E,q) = 0
\end{equation}     
    is equivalent to
\begin{equation}\label{eq:critequiv2}
 \frac{E}{E_{\infty}}=\frac{1}{2}\left(\frac{1}{\sqrt{2}\beta_{2}\left(q\right)}+\sqrt{2}\beta_{2}\left(q\right)\right),
\end{equation}
and also to
\begin{equation}\label{eq: critequiv r}
 E\ge E_{\infty}\text{ and }\begin{cases}
\sqrt{2}\beta_{2}\left(q\right)=r_{-}\left(E/E_\infty\right) & \text{or }\\
\sqrt{2}\beta_{2}\left(q\right)=r_{+}\left(E/E_\infty\right).
\end{cases}
\end{equation}
\end{lem}
\begin{proof}
Taking the derivative in $q$ of \eqref{eq: f def} we get
\begin{equation}\label{eq: f first deriv}
     \partial_q f(E,q) =\frac{\beta p}{2} q^{p/2-1}E - \frac{1}{2(1-q)}-\frac{\beta^2}{2}\xi''(q)(1-q).
\end{equation}
Using the definition \eqref{eq: beta2 def} of $\beta_2(q)$ we obtain
\begin{equation}\label{eq: f first deriv div}
     \partial_q f(E,q) = \frac{\sqrt{2}\beta_2(q)}{1-q} \left\{
     \frac{p}{2\sqrt{\xi''(q)}} q^{p/2-1}E- \frac{1}{2}\left(\frac{1}{\sqrt{2}\beta_{2}\left(q\right)}+\sqrt{2}\beta_{2}\left(q\right)\right)
     \right\}.
\end{equation}
We have
\begin{equation}\label{eq: xi double prime}
    \sqrt{\xi(q)} \overset{\xi(x)=x^p}{=} \sqrt{p(p-1)}q^{\frac{p-2}{2}},
\end{equation}
so that
\begin{equation}\label{eq: Einf appears}
\frac{ p}{2 \sqrt{\xi''(q)}} q^{p/2-1}
= \frac{\sqrt{p}}{2\sqrt{p-1}} \overset{\eqref{eq: Einf def}}{=} \frac{1}{E_\infty},
\end{equation} 
implying \eqref{eq: f deriv}.
The equivalence of \eqref{eq: crit point eq} and \eqref{eq:critequiv2} follows,
and the equivalence to \eqref{eq: critequiv r} follows by \eqref{eq: rpmdef as sols}.
\end{proof}
\begin{rem} Note that we do not use anything about the complexity $I$ of the critical points of $H_N(\sigma)$ to obtain \eqref{eq:critequiv2}, but nevertheless a numerical value which can be written as the threshold  $E_\infty$ arising from $I$ appears in \eqref{eq: Einf appears}.
\end{rem}
To count the number and location of critical points of $q \to f(E,q)$ we should thus count solutions of \eqref{eq: critequiv r}. To this end note that from the definition \eqref{eq: beta2 def} of $\beta_2$ and \eqref{eq: xi double prime}
\begin{equation}\label{eq: beta2 concrete form}
    \beta_2(q) = \beta \sqrt{\frac{p(p-1)}{2}} (1-q) q^{\frac{p-2}{2}}.
\end{equation} 
From this one easily checks that $\beta_2(0) = \beta_2(1) = 0$, and (by considering its derivative) that \begin{equation}\label{eq: beta2 shape}
    \beta_2(q)\text{ is strictly increasing on } \left[0,\frac{p-2}{p}\right]\text{, strictly decreasing on  }\left[\frac{p-2}{p},1\right],
\end{equation}
and maximized at $q=\frac{p-2}{p}$.
With this knowledge we note that
\begin{equation}\label{eq: w def}
\sqrt{2} \sup_{q\in[0,1]} \beta_2\left(q\right) =  \sqrt{2}\beta_2\left(\frac{p-2}{p}\right) \overset{\eqref{eq: beta tilde c def}, \eqref{eq: beta2 concrete form}}{=} \frac{\beta}{\tilde{\beta}_c},
\end{equation}
and
\begin{equation}\label{eq: num sols beta2}
\begin{array}{l}
\text{for fixed }a\in\mathbb{R}\text{\ the equation }\sqrt{2}\beta_{2}\left(q\right)=a\text{\,has }\\
\quad\begin{cases}
\text{no solutions } & \text{\,if }a>\betatwosup,\\
\text{exactly one solution, namely }q=\frac{p-2}{p}, & \text{\,if }a=\betatwosup,\\
\text{exactly two solutions, one in }\left(0,\frac{p-2}{p}\right)\text{\,and one in }\left(\frac{p-2}{p},1\right), & \text{\,if }a<\betatwosup.
\end{cases}
\end{array}
\end{equation}

Using mainly the form \eqref{eq: critequiv r} of the critical point equation and \eqref{eq: num sols beta2} we now show that $q\to f(E,q)$ in general has between $0$ and $4$ critical points, of which up to $2$ can be local maxima. When there are several local maxima it turns out that at most one satisfies the condition $q \in D_\beta$, and thus at most one can correspond to a relevant ($g$-)TAP solution. 
Some of the possible cases are illustrated in Figure \ref{fig:plots_of_f}.
The complete result is the following.

\begin{teor}\label{thm: det char f}
In the following statement ``critical point of $f$'' always refers to a critical point of $q \to f(E,q)$ in the interval $[0,1]$, for fixed $E$.
\begin{enumerate}
\item If  $E < E_\infty$ then $f$ is decreasing and has no critical points.
\item If $E>E_\infty$ the following holds:
\begin{enumerate}
\item If
$\betatwosup<r_{-}\left(E/E_\infty\right)$
then $f$ has no critical points.
\item
If $\betatwosup=r_{-}\left(E/E_\infty\right)$
then $f$ has one critical point, namely a saddle point at $q= \frac{p-2}{p}$.
\item If $r_{-}\left(E/E_\infty\right)<\betatwosup<r_{+}\left(E/E_\infty\right)$
then $f$ has two critical points: a local maximum in $D_\beta \cap  (\frac{p-2}{p},1)$ and a local minimum in $(0,\frac{p-2}{p})$.
\item If $\betatwosup = r_{+}\left(E/E_\infty\right)$ then $f$ has three critical points: a local maximum in $D_\beta \cap (\frac{p-2}{p},1)$, a local minimum in $(0,\frac{p-2}{p})$ and a saddle point at $q=\frac{p-2}{p}$.
\item If $\betatwosup>r_{+}\left(E/E_\infty\right)$ then $f$ has four critical points: in $(0,\frac{p-2}{p})$ a local maximum outside $D_\beta$ and a local minimum, and in $(\frac{p-2}{p},1)$ a local maximum inside $D_\beta$ and a local minimum outside $D_\beta$.
\end{enumerate}

\item In the special case $E = E_\infty$ for which $ r_{+}\left(E/E_\infty\right)= r_{-}\left(E/E_\infty\right) =  1$, the function $f$ is decreasing and has (a) no critical points if $ \betatwosup < 1$, (b) a single critical point at $q=\frac{p-2}{p}$ which is a saddle point if $ \betatwosup = 1$ and (c) exactly two critical points, one saddle point in $(\frac{p-2}{p},1)$ and one saddle point in $(0,\frac{p-2}{p})$, if $ \betatwosup > 1$.
\end{enumerate}
\end{teor}
   
\begin{figure}[H]
  \centering 
   \includegraphics[width= 0.325\columnwidth]{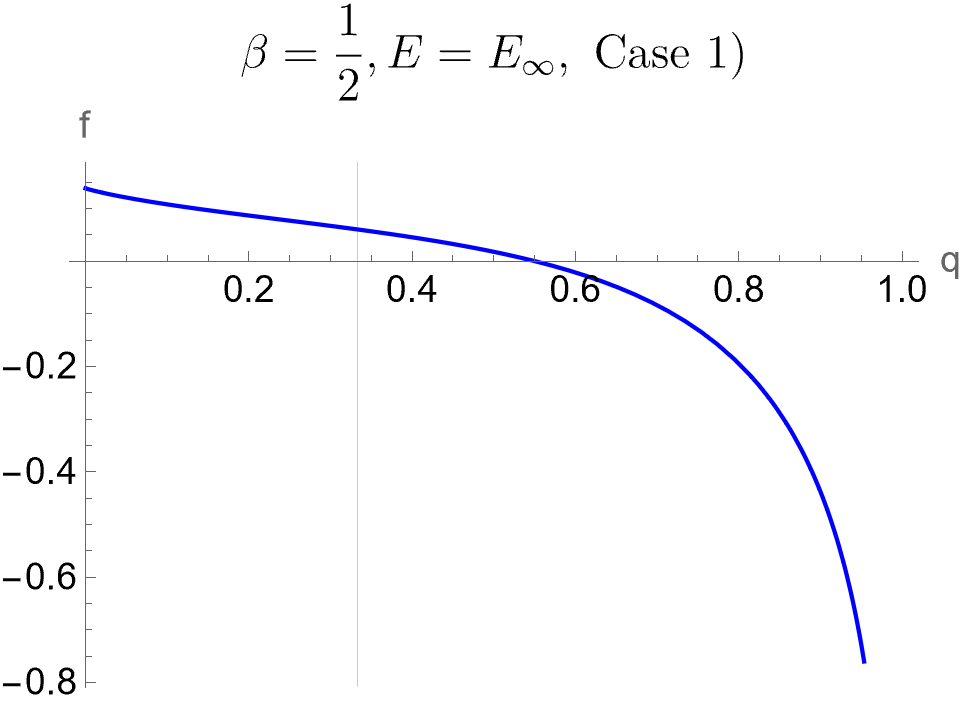}
   \includegraphics[width= 0.325\columnwidth]{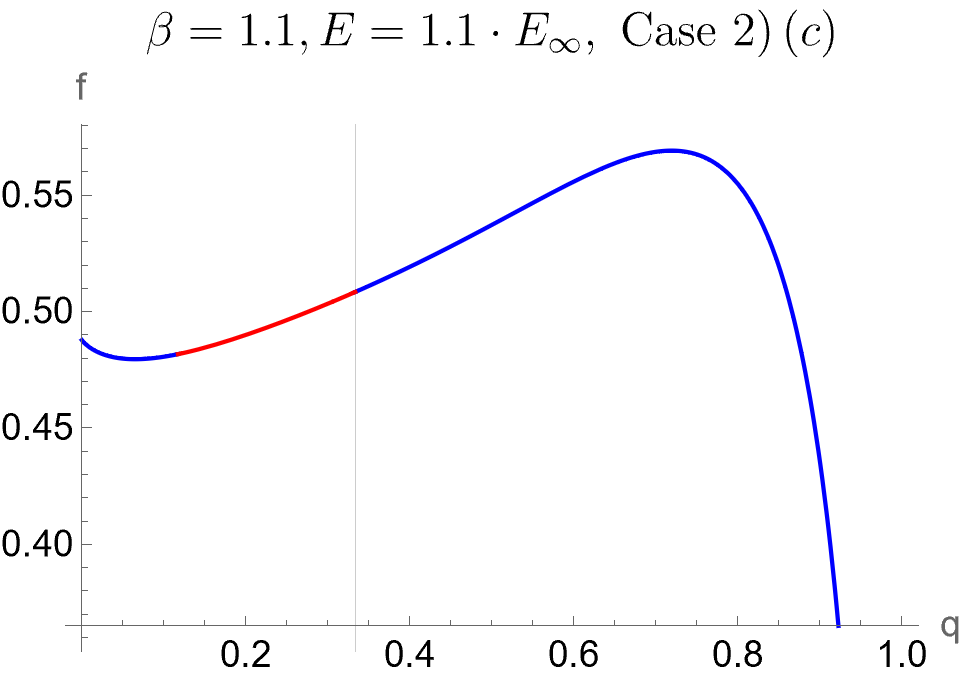}
   \includegraphics[width= 0.325\columnwidth]{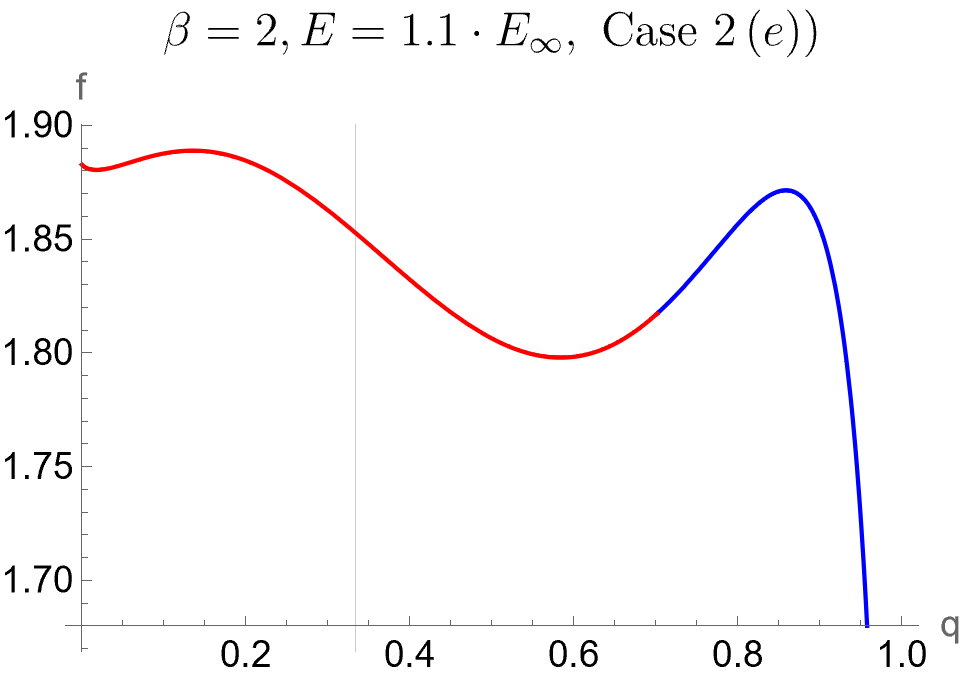}
  \caption{Plots of $f(E,q)$ in $q$ with $p=3$, giving examples for cases 1), 2) (c) and 2) (e). The graph is blue if $q\in D_\beta$ and red otherwise. Horizontal line at $q=\frac{p-2}{p}$. }\label{fig:plots_of_f}
\end{figure}

As seen above, for certain values of $E$ and $\beta$ the function $q\to f(E,q)$
may have critical points that are not local maxima or do not satisfy
the condition $q \in D_\beta$. These physically non-relevant critical points
give rise to physically non-relevant TAP solutions.

\begin{proof}[Proof of Theorem \ref{thm: det char f}]
Recall Lemma \ref{lem: critequiv}. Since the RHS of \eqref{eq:critequiv2} is always at least $1$ it follows that $f$ is decreasing and there are no critical points if $E<E_\infty$, proving 1).

Turning to 2), note that the number of critical points is the number of unique solutions to \eqref{eq: critequiv r}. If $E=E_\infty$ then $r_{-}\left(E/E_\infty\right)=r_{+}\left(E/E_\infty\right)$ and these are actually one equation, and
\begin{equation}\label{eq: rpm diff from one}
r_{-}\left(E/E_\infty\right) < 1 < r_{+}\left(E/E_\infty\right) \text{ if } E>E_\infty,
\end{equation}
so then the equations are distinct. In the latter case it holds thanks to \eqref{eq: num sols beta2} that there are no solutions if $\betatwosup < r_{-}\left(E/E_\infty\right)$, one solution if $\betatwosup = r_{-}\left(E/E_\infty\right) <r_{+}\left(E/E_\infty\right)$, two solutions if $ r_{-}\left(E/E_\infty\right)< \betatwosup <r_{+}\left(E/E_\infty\right)$, three solutions if  $ r_{-}\left(E/E_\infty\right)<  r_{+}\left(E/E_\infty\right) = \betatwosup$ and four solutions if $ r_{-}\left(E/E_\infty\right)<  r_{+}\left(E/E_\infty\right) < \betatwosup$. The fact \eqref{eq: num sols beta2} also gives information on if these critical points belong to $(0, \frac{p-2}{p})$ or $(\frac{p-2}{p},1)$, or equal $\frac{p-2}{p}$. Furthermore Lemma \ref{lem:ourplefkaprofs} 1) 2) and the fact that $r_{-}\left(E/E_\infty\right)<1<r_{+}\left(E/E_\infty\right)$ when $E>E_\infty$ implies that no critical point arising from $\sqrt{2} \beta_2(q)=r_{+}\left(E/E_\infty\right)$ is ever in $D_\beta$, and all critical points in $(\frac{p-2}{p},1)$ that arise from $\sqrt{2} \beta_2(q)=r_{-}\left(E/E_\infty\right)$ lie in $D_\beta$. In this way all claims about the number and location (but not index) of critical points in 2 a)-e) follow.

The claims about number and location of critical points in 3) similarly follows keeping in mind that if $E=E_\infty$ then the two equations in  \eqref{eq: critequiv r} coincide.

All claims about number and positions of the critical points in claims 1)-3) are thus proven. To conclude the proof it only remains to determine if the critical points are local maxima, local minima or saddle points. Differentiating \eqref{eq: f first deriv} and using that 
$$\xi'''\left(q\right)\left(1-q\right)-\xi''\left(q\right)=p\left(p-1\right)q^{p-3}\left(\left(p-2\right)\left(1-q\right)-q\right)$$
one gets
\begin{equation}\label{eq: f second deriv}
    \partial_{q}^{2}f\left(E,q\right)=\beta\frac{p\left(p-2\right)}{4}q^{\frac{p}{2}-2}E-\frac{1}{2\left(1-q\right)^{2}}-\frac{\beta^{2}}{2}p\left(p-1\right)q^{p-3}\left(\left(p-2\right)\left(1-q\right)-q\right).
\end{equation}
We now show that at a solution to $\partial_{q}f\left(E,q\right)=0$ this factors as
\begin{equation}\label{eq: factorization of f sec deriv}
    \partial_{q}^{2}f\left(E,q\right)=\frac{p}{4q\left(1-q\right)^{2}}\left(q-\frac{p-2}{p}\right)\left(2\beta_{2}\left(q\right)^{2}-1\right).
\end{equation}
To see this, note that using \eqref{eq: beta2 concrete form} and \eqref{eq: Einf def} we can make $\beta_{2}\left(q\right)$ appear in the first and last terms of \eqref{eq: f second deriv}, and $E_{\infty}$ appear in the first term, obtaining
$$\partial_{q}^{2}f\left(E,q\right)=\frac{p-2}{2q\left(1-q\right)}\sqrt{2}\beta_{2}\left(q\right)\frac{E}{E_{\infty}}-\frac{1}{2\left(1-q\right)^{2}}-\beta_{2}^{2}\left(q\right)\frac{\left(p-2\right)\left(1-q\right)-q}{q\left(1-q\right)^{2}}.$$
Using that at a solution to $\partial_{q}f\left(E,q\right)=0$ the equality \eqref{eq:critequiv2} holds to remove $\frac{E}{E_\infty}$ from the expression, we get that at such a point
\begin{equation}
    \begin{array}{rcl}
        \partial_{q}^{2}f\left(E,q\right)
        &=&\displaystyle{\frac{\left(p-2\right)\left(1-q\right)-2q + 2\beta_{2}^{2}\left(q\right)\Big(\left(p-2\right)\left(1-q\right)-2( \left(p-2\right)\left(1-q\right)-q) \Big)}{4 q\left(1-q\right)^{2}}}\\
        &=&\displaystyle{\frac{\left( (p-2)(1-q)-2q \right) \left( 2\beta^2(q)-1\right)}{4 q\left(1-q\right)^{2}}},
    \end{array}
\end{equation}
so since $\left(p-2\right)\left(1-q\right)-2q = \left(p-2\right)-pq = -p \left(q-\frac{p-2}{p}\right)$
we obtain that \eqref{eq: factorization of f sec deriv} holds at any critical point $q$.

By \eqref{eq: rpm diff from one} solutions of $\sqrt{2} \beta_2(q) = r_{-}\left(E/E_\infty\right)$ satisfy $2\beta_2(q)^2<1$ when $E>E_\infty$, so that by checking the sign of \eqref{eq: factorization of f sec deriv} any critical point arising from that equation in $(0,\frac{p-2}{p})$ is a local minimum, and any such critical point in $(\frac{p-2}{p},1)$ is a local maximum. Similarly any solutions of $\sqrt{2} \beta_2(q) = r_{+}\left(E/E_\infty\right)$ satisfies $2\beta_2(q)^2>1$ when $E>E_\infty$, so that any critical point arising from that equation in $(0,\frac{p-2}{p})$ is a local maximum, and any such critical point in $(\frac{p-2}{p},1)$ is a local minimum. This concludes the identification of all claimed local maxima and minima in 2).

 It remains to prove that in the remaining cases the critical points are saddle points. 
When $E=E_\infty$ we have from \eqref{eq: f deriv}
\[ \partial_q f(E,q) =-\frac{(\sqrt{2}\beta_2(q)-1)^2}{2(1-q)}.  \]
As this is non-positive and only touches but never crosses $0$ at a finite number of points $f$ is decreasing and all critical points for $E=E_\infty$ are saddle points, concluding the proof of claim 3).

Next in the cases 2) (b) (d) recall that all unclassified critical points are at $q=\frac{p-2}{p}$, so that at these points
\beq
\frac{E}{E_{\infty}}=\frac{1}{2}\left(\frac{1}{\sqrt{2}\beta_2\left(\frac{p-2}{p}\right)}+\sqrt{2}\beta_2\left(\frac{p-2}{p}\right)\right),
\eeq
giving us from \eqref{eq: f deriv} and the identity $\frac{1}{x}+x - (\frac{1}{y}+y) = \frac{1}{y}(\frac{1}{x}-y)(y-x)$ that
\[ \partial_q f(E,q) =\frac{\left(\frac{1}{\sqrt{2}\beta_2\left(\frac{p-2}{p}\right)}-\sqrt{2}\beta_2(q)\right)\left(\sqrt{2}\beta_2(q)-\sqrt{2}\beta_2\left(\frac{p-2}{p}\right)\right)}{2(1-q)}.  \]
The first factor has the same non-zero sign throughout $(\frac{p-2}{p}-\varepsilon,\frac{p-2}{p}) \cup (\frac{p-2}{p},\frac{p-2}{p}+\varepsilon)$ for some small enough $\varepsilon$ (the midpoint $q=\frac{p-2}{p}$ can also be included if $\beta \ne \tilde{\beta}_c$; when $\beta = \tilde{\beta}_c$ the first factor is zero there) while the second only touches zero (not crossing), since $q=\frac{p-2}{p}$ maximizes $\beta_2$. Hence $q=\frac{p-2}{p}$ is a saddle point, concluding the proof of claims 2) (b) (d).
This concludes the proof of 1)-3).

\end{proof}
We will use the following consequences of the theorem.
\begin{cor}\label{cor: conseq}
The following holds for all $E$.
\begin{enumerate}
\item 
There is at most one local maximum of $q \to f(E,q)$ in $D_\beta$.

\item 

All local maxima of $q \to f(E,q)$ that lie in $D_\beta \setminus \{0\}$ also lie in $(\frac{p-2}{p},1)$.

\item 

All local maxima in $(\frac{p-2}{p},1)$ satisfy 
$\sqrt{2}\beta_2(q) = r_{-}\left(E/E_\infty\right)$.

\item When it exists, the unique local maximum in $D_\beta \cap ( \frac{p-2}{p},1)$ is the global maximum of $q \to f(E,q)$ in $D_\beta \cap ( \frac{p-2}{p},1)$.

\end{enumerate}
\end{cor}
\begin{proof}
1)-3) follows directly by examining all the possible cases in 1)-3) in the previous theorem. The claim 4) follows since if a differentiable function has only one local maximum and no minima in an interval then this local maximum is the global maximum in the interval.
\end{proof}

Lemma \ref{lem: det TAP sol radial sym}, Theorem \ref{thm: det char f} and Corollary \ref{cor: conseq} strongly constrain which combinations of energy $\HNg(\mhattext)$ at local maximum, norm $|m|^2$ of \gTAP\, solution and \gTAP\, energy $\HTAPg(m)$ are possible for a relevant \gTAP\, solution.
Let
\begin{equation}
E_{\min}=\begin{cases}
\frac{E_{\infty}}{2}\left( \frac{\tilde{\beta}_{c}}{\beta}+\frac{\beta}{\tilde{\beta}_{c}}\right)\ge E_{\infty} & \text{ if } \beta \le \tilde{\beta}_{c},\\
E_{\infty} & \text{ if } \beta \ge \tilde{\beta}_{c}.
\end{cases}\label{eq: Emin def}.
\end{equation}
Note that
\begin{equation}\label{eq: Emin iff}
    \begin{array}{rcl}
        E > E_{\infty} \text{ and }r_{-}(E/E_\infty) > \betatwosup & \iff & E>E_{\min},\\
        E < E_{\infty} \text{ or }\left(E\ge E_{\infty} \text{ and }r_{-}(E/E_\infty)< \betatwosup\right)  & \iff & E<E_{\min}
    \end{array}
\end{equation}
(recall \eqref{eq: rpmdef as sols}).
Thus if $m$ is a relevant TAP solution then by Theorem \ref{thm: det char f}
the energy $E=g(\hat{m})$ satisfies $E>E_{\min}$. Also define
\begin{equation}\label{eq: qmin def}
q_{\min}=\begin{cases}
\frac{p-2}{p} & \text{ if }\beta \le \tilde{\beta}_c,\\
\text{unique solution of }\sqrt{2}\beta_{2}\left(q\right)=1\text{ in }[\frac{p-2}{p},1) & \text{ if }\beta \ge \tilde{\beta}_c.
\end{cases}
\end{equation}
Note that by Lemma \ref{lem:ourplefkaprofs} a) and \eqref{eq: beta2 concrete form}-\eqref{eq: beta2 shape} we have
\begin{equation}\label{eq: Dbeta is interval in dyn low temp}
D_\beta \cap \left[ \frac{p-2}{p},1\right]
 = [q_{\min},1].
\end{equation}
Later we will use that
\begin{equation}\label{eq: qmin inc}
    q_{\min} \text{ is strictly decreasing in } \beta \text{ when } \beta\ge \tilde{\beta}_c,
\end{equation}
(recall \eqref{eq: qmin def}, \eqref{eq: beta2 concrete form}, \eqref{eq: beta2 shape}) and
\begin{equation}\label{eq: qmin at betad}
    q_{\min} = \frac{p-2}{p-1} \text{ when } \beta = \beta_d,
\end{equation}
since when $\beta=\beta_d$ and $\hat{q}=\frac{p-2}{p-1}$ we have
$$\sqrt{2}\beta_2(\hat{q}) \overset{\eqref{eq: beta2 concrete form}}{=}\beta_d \sqrt{p(p-1)} \left(1 - \hat{q}\right)\hat{q}^{\frac{p-2}{2}} \overset{\eqref{eq: beta d def}}{=}1.$$

Define the function
\begin{equation}\label{eq: Eq def}
E_q:[q_{\min},1)\to [E_{\min},\infty) \text{ by } E_q(q) = \frac{E_\infty}{2}\left( \frac{1}{\sqrt{2}\beta_2(q)}+\sqrt{2}\beta_2(q)\right),
\end{equation}
cf. \eqref{eq:critequiv2}.
By \eqref{eq: w def} we have $E_q(q_{\min})=E_{\min}$ and by \eqref{eq: beta2 shape} and since $\sqrt{2}\beta_2(q)\le1$ for $q\ge q_{\min}$ (see \eqref{eq: Emin iff} and \eqref{eq: w def})
\begin{equation}\label{eq: Eq increasing}
    E_q\text{ is strictly increasing}.
\end{equation}
Therefore we can define a function
\begin{equation}\label{eq: qE def}
    q_E:[E_{\min},\infty)\to[q_{\min},1) \text{ by } q_E = E_q^{-1},
\end{equation}
for which
\begin{equation}\label{eq: qE increasing}
    q_E\text{ is strictly increasing}.
\end{equation}
There are several useful ways to characterize $q_E$. Since \eqref{eq:critequiv2} and \eqref{eq: critequiv r} are equivalent we have by \eqref{eq: num sols beta2} that
\begin{equation}\label{eq: qE as sol not conc}
    q_E(E) \text{ is the unique solution to } \sqrt{2}\beta_2(q)=r_{-}(E/E_\infty) \text{ in } \left[\frac{p-2}{p},1\right],
\end{equation}
or equivalently using \eqref{eq: beta2 concrete form} that
\begin{equation}\label{eq: qE as sol conc}
    q_E(E) \text{ is the unique solution to } (1-q) q^{\frac{p-2}{p}} = \frac{r_{-}(E/E_\infty)}{\beta\sqrt{p(p-1)}} \text{ in } \left[\frac{p-2}{p},1\right].
\end{equation}
Alternatively we have the following.
\begin{lem}For all $\beta\ge 0,E\ge E_{\min}$
\begin{equation}\label{eq: qE glob max}
q_{E}\left(E\right)\text{ is the unique critical point and global maximum of }q\to f\left(E,q\right)\text{ in }D_{\beta} \cap \left[\frac{p-2}{p},1\right],
\end{equation}
and
\begin{equation}\label{eq: qE loc max}
    q_E(E)\text{ is a local maximum iff }E>E_{\min}.
\end{equation}
\end{lem}
\begin{proof}
By \eqref{eq: qE as sol not conc}, the equivalence of \eqref{eq: crit point eq} and \eqref{eq: critequiv r} and the fact that no solution to $\sqrt{2}\beta_2(q)=r_{+}(E/E_\infty)$ can lie in $D_\beta$ it follows that $q_E(E)$ is the unique critical point in the interval. 
By examining all cases in Theorem \ref{thm: det char f} and recalling \eqref{eq: Emin iff} we get \eqref{eq: qE loc max}. By Corollary \ref{cor: conseq} 4) the claim \eqref{eq: qE glob max} thus follows for $E>E_{\min}$.

The special case $E=E_{\min}$ follows since then $q_{\min}=q_E(E_{\min})$ is a critical point of $q\to f(E,q)$ by \eqref{eq:critequiv2}, \eqref{eq: Eq def} and \eqref{eq: qE def}, which by Theorem \ref{thm: det char f} 3) is a saddle point, and is also is the left-endpoint of $D_{\beta} \cap \left[\frac{p-2}{p},1\right]$, and $f(E,q)\to-\infty$ for $q\to1$, so that the saddle point is the maximum.
\end{proof}

We also define
\begin{equation}\label{eq: Uminbeta}
    U_{\min}=f\left(E_{\min},\qE\left(E_{\min}\right)\right),
\end{equation}
and the function
\begin{equation}\label{eq: def UE-1}
\UE\left(E\right)=f\left(E,\qE\left(E\right)\right).
\end{equation}
The fact that $f\left(E,q\right)$ is strictly increasing in $E$ (see \eqref{eq: f def}), $q_{\min}\ge\frac{p-2}{p}$ and \eqref{eq: qE glob max} implies that 
\begin{equation}
\UE\text{\,is strictly increasing}.\label{eq: UE strictly increasing}
\end{equation}
Therefore there is a function
 \begin{equation}
 \EU:[U_{\min},\infty)\to[E_{\min},\infty) \text{ defined by }\EU=\UE^{-1},\label{eq: final def EU}
 \end{equation}
 and
\begin{equation}\label{eq: EU strictly increasing}
     \EU\text{\,is strictly increasing}.
\end{equation}
We have the following.
\begin{lem}\label{lem: TAP sol iff} A vector $m$ is a relevant non-zero \gTAP\ solution of energy $U$ iff $\hat{m}$ is a local maximum of $g$, $U > U_{\min}$, $g(\hat{m}) = E_U(U)$ and $|m|^2=q_E(E_U(U))$.
\end{lem}
\begin{proof}
By Lemma \ref{lem: det TAP sol radial sym} 1), Theorem \ref{thm: det char f},  \eqref{eq: Emin iff} and \eqref{eq: qE glob max}-\eqref{eq: qE loc max} we have that a vector $m$ is a relevant non-zero TAP solution with $g(\hat{m})=E$ iff $\hat{m}$ is a local maximum of $g$, $E>E_{\min}$ and $|m|^2=q_E(E)$. Since $H^g_{\TAP}(m) = U_E(g(\hat{m}),|m|^2)$ we get the claim with the bijective change of variables $U=U_E(E)$ and \eqref{eq: Uminbeta}.
\end{proof}
\begin{rem}
The above lemma implies that (if $g$ is random) there are no relevant g-TAP solutions $m$ of energy $\frac{1}{N}H^g_{\rm{TAP}}(m)\le U_{\min}$ almost surely.
\end{rem}
Theorem \ref{thm: det char f} and \eqref{eq: Emin iff} also imply the next lemma.
\begin{lem}\label{lem: E g Emin loc max}
If $E>E_{\min}$ then $q\to f(E,q)$ has only one critical point in $D_\beta$, which is a local maximum.
If $E<E_{\min}$ then $q\to f(E,q)$ has no critical points in $D_\beta$.
\end{lem}

\section{Complexity threshold}\label{sec: comp thres}
In this section we prove Theorem \ref{thm:comp} about the complexity threshold, using the results of the previous section and the complexity of critical points from \eqref{eq: I props}.
\begin{proof}[Proof
of Theorem \ref{thm:comp}]
We first prove \eqref{eq:onlyzerosol no plefka}. This implies also implies \eqref{eq:onlyzerosol} since by Lemma \ref{lem: det TAP sol radial sym} 2) $m=0$ is a relevant TAP solution almost surely when $\beta \le \beta_d$, so in particular it is when $\beta \le \beta_c < \beta_d$. 

By Lemma \ref{lem: det TAP sol radial sym} 1), \eqref{eq: crit point eq}-\eqref{eq:critequiv2} and \eqref{eq: w def} when $\beta \le \tilde{\beta}_c$ any non-zero TAP solution must satisfy
\begin{equation}\label{eq: tap sol must sat}
    H_N(\hat{m})\ge\frac{E_\infty}{2}\left( \frac{\tilde{\beta}_{c}}{\beta}+\frac{\beta}{\tilde{\beta}_{c}} \right).
\end{equation}
Note that when $\beta \le \tilde{\beta}_c$
\begin{equation}\label{eq: betac appears}
\frac{E_{\infty}}{2}\left( \frac{\tilde{\beta}_{c}}{\beta}+\frac{\beta}{\tilde{\beta}_{c}} \right)>E_{0}\iff\frac{1}{2}\left(\frac{\tilde{\beta}_{c}}{\beta}+\frac{\beta}{\tilde{\beta}_{c}}\right)>\frac{E_{0}}{E_{\infty}} \overset{\eqref{eq: rpmdef as sols}}{\iff} \frac{\beta}{\tilde{\beta}_{c}}<r_{-}\left(\frac{E_0}{E_\infty}\right)
\overset{\eqref{eq: rbardef},\eqref{eq: betacdef}}{\iff}\beta<\beta_{c}.
\end{equation}
Thus the claim \eqref{eq:onlyzerosol} follows since \eqref{eq: ground state} implies that the probability of an $\hat{m}$ satisfying \eqref{eq: tap sol must sat} existing goes to zero.
The claims \eqref{eq:Itapbeforebetac} and \eqref{eq: FE very high temp} are simple consequences of \eqref{eq:onlyzerosol} and the definitions \eqref{eq: def TAP comp} of $I_{\rm TAP}$ and  \eqref{eq:varprinc} of $\fTAP(\beta)$.

Conversely if $\beta > \beta_c$ then we have that
$$ r_{-}\left(\frac{E}{E_\infty}\right) < \frac{\beta}{\tilde{\beta}_c} \text{ for } E\in[E_0-\delta,E_0],$$
for any $\delta>0$.
By Theorem \ref{thm: det char f} 2) and \eqref{eq: qE glob max} the function $q\to f(E,q)$ thus has a local maximum $q_E$ with $q_E \in D_\beta \setminus \{0\}$ for all $E\in[E_0-\delta,E_0]$. This means that if $\hat{m}$ is a critical point of $H_N$ with $\frac{1}{N}H_N(\hat{m}) \in [E_0-\delta,E_0]$ then $m = q^{p/2} \hat{m}$ is a relevant TAP solution. Thus using the notations from \eqref{Ndef} and above \eqref{eq: loc max quenched}
$$ \mathcal{N}_N(\mathbb{R}, \mathbb{R}, D_\beta \setminus \{0\}) \ge \mathcal{M}([E_0-\delta,E_0)).$$
Taking logs and dividing by $N$ implies \eqref{eq: exp many TAP sol} by the definitions \eqref{eq: def TAP comp} and \eqref{eq: loc max quenched}, since $I( [E_0-\delta,E_0))>0$ for all $\delta>0$ (recall \eqref{eq: I props}).
\end{proof}
\begin{rem}\label{rem: betac}
Eq. \eqref{eq: betac appears} explains the origin of the threshold $\beta_c$: it is the first $\beta$ where critical points of $H_N$ of energy as low as $N E_0$ give rise to relevant TAP solutions.
\end{rem}

\section{Computation of the TAP rate function}\label{computeI}
In this section we give a more detailed version of Theorem \ref{thm: TAP comp intro} about the TAP rate function.
Define
\begin{equation}\label{eq: Uq def}
    U_q(q) = f(E_q(q),q) = \beta q^{p/2} E_q(q) +  \frac{1}{2}\log(1-q) + \On(q).
\end{equation}
Since $x\to x+1/x$ is strictly decreasing for $x\le1$ and $q\to \sqrt{2} \beta_2(q)$ is strictly decreasing for $q \ge q_{\min}$ (recall \eqref{eq: qmin def}) we have that $E_q$ is strictly increasing for such $q$. Thus since $U_q'(q) = \beta q^{p/2} E_q'(q) + \frac{d}{dq} f(E_q(q),q) = \beta q^{p/2} E_q'(q)>0$ for $q\ge q_{\min}$ we have that
$$ U_q(q) \text{ is strictly increasing for } q\ge q_{\min} .$$
From \eqref{eq: Uq def}, \eqref{eq: def UE-1} and \eqref{eq: qE def} we have the natural relation \begin{equation}\label{eq: Uq natural rel}
    U_q(q)=U_E(E_q(q)).
\end{equation}
Recall that from \eqref{eq: Uminbeta}
\begin{equation}\label{eq: Umax Umin def}
U_{\min}= U_E(E_{\min}) = \begin{cases}
U_E(E_{\min}) & \text{\,if }\beta_{c}\le\beta\le\tilde{\beta}_{c},\\
U_E(E_\infty) & \text{\,if }\beta\ge\tilde{\beta}_{c},
\end{cases}
\text{ and let }
U_{\max}=\begin{cases}
-\infty & \text{\,if }\beta<\beta_{c},\\
U_E(E_0) & \text{\,if }\beta\ge\beta_{c}.
\end{cases}
\end{equation}
Since $E_U=U_E^{-1}$ it follows trivially from these that
\begin{equation}\label{eq: EU of Umin and Umax}
    E_U(U_{min}) = E_{\min} \text{ and } E_U(U_{\max})=E_0.
\end{equation}
In Proposition \ref{prop: umax umin formulas} we give more concrete formulas for $U_{\min}$ and $U_{\max}$. Our full result on the TAP complexity is the following. There are many ways to express the dependence of the $\ITAP$ on $I$; we choose to present a verbose version and a compact version.
\begin{teor}[TAP entropy in terms of critical point entropy]\label{thm:TAP comp}
For all $\beta\ge \beta_c$ we have $U_{\min}\le U_{\max}$, with equality only if $\beta = \tilde{\beta}_c$. Furthermore it holds that
\begin{equation}\label{eq: ITAP verbose}
I_{{\rm TAP}}\left(U,V,q\right)=\begin{cases}
I(E_{\min}) \overset{\text{for } \beta > \tilde{\beta}_{c}}{>}0 & \text{ if }U=U_{\min}\\
& \quad\text{and }q=q_E(E_{\min}), V=q^{p/2} E_{\min},\\
I\left(E\right) \in (0,I(E_{\min})) & \text{ if }U_{\min} < U< U_{\max}\\
 & \quad\text{and }q\text{\,is the unique solution to   }U_{q}\left(q\right)=U\text{ in }[q_{\min},1),\\
 & \quad\text{and } V= q^{p/2} E \text{ and }E=E_q(q) \in (E_{\min},E_0)\\
I(E_0)=0 & \text{if }U=U_{\max}\\
&\quad\text{and }q=q_E(E_0), V=q^{p/2} E_0,\\
0 & \text{\,if }U=\frac{\beta^{2}}{2}\text{ and } \beta\le\beta_{d}\\
&\quad\text{and }V=0, q=0,\\
-\infty & \text{otherwise}.
\end{cases}
\end{equation}
Alternatively it holds for all $\beta\ge 0$ and $U$ that
\begin{equation}\label{eq: TAP comp}
\ITAP\left(\ETAP,E,q\right)= 
\begin{cases}
I\left(E_{U}\left(U\right)\right) & \text{ if } U\ge U_{\min}, q = q_E(E_U(U)), V=q^{p/2} E_U(U),\\
0 & \text{\,if }U=\frac{\beta^{2}}{2}, V=0, q=0\text{ and }\beta\le\beta_{d}\\
-\infty & \text{ otherwise}.
\end{cases}
\end{equation}
\end{teor}
\begin{rem}\label{rem: ITAP thm rems} a) From \eqref{eq: TAP comp} one sees that for each $U$ there is at most one $(V,q)$ such that $\ITAP(U,V,q)\ne-\infty$.\\

b) From the first three cases in \eqref{eq: ITAP verbose} one sees that \TP solutions of minimal energy $U_{\min}$ are the most numerous, and they have complexity $I(E_{\min})$ ($=I(E_\infty)$ if $\beta \ge \tilde{\beta}_c$, while for any $U>U_{\min}$ the complexity is $I(E)$ for some $E>E_{\min}$ and $I(E)<I(E_{\min})$). Also from the third case one sees that the number of \TP\ solutions within $o(N)$ of the maximal TAP energy is subexponential, since their complexity $0$.\\

c) Furthermore combined with \eqref{eq: Emin def} wee see the meaning of the threshold $\tilde{\beta}_c$: for $\beta>\tilde{\beta}_c$ all critical points of $H_N$ on $S_{N-1}$ of any energy in $[E_\infty,E_0]$ give rise to \TP\ solutions, while for $\beta\in (\beta_c,\tilde{\beta}_c)$ we have $E_{\min} > E_\infty$ and only critical points with energies in $[E_{\min},E_0]$ do so.
\end{rem}
\begin{proof}
Let
\begin{equation}\label{eq: qU def}
    q_U(U)=q_E(E_U(U))
\end{equation}
and note that $q_U: [U_{\min}, \infty)\to[q_{\min},1)$ is strictly increasing (see \eqref{eq: qE increasing} and \eqref{eq: EU strictly increasing}), and $q_U=U_q^{-1}$ since by definition and \eqref{eq: Uq natural rel} we have $q_U(U_q(q))=q_E(E_U(U_q(q))=q$. Also let $V_U(U) = q_U(U)^{p/2} E_U(U)$. Since $q_U$ and $E_U$ are increasing so is $V_U$. Now \eqref{eq: TAP comp} follows essentially directly since by Lemma \ref{lem: TAP sol iff} an $m\ne0$ is a \TP\ solution of energy $U$ iff $\hat{m}$ is a critical point of $H_N$ such that $\frac{1}{N}H_N(\hat{m})=E_U(U)$, $\frac{1}{N}H_N(m) = V_U(U)$ and $|m|^2=q_U(U)$.

The detailed argument is the following. Recall that $E_U,V_U,q_U$ are strictly increasing. Furthermore they are continuous and differentiable ($E_q$ is by \eqref{eq: Eq def} and \eqref{eq: beta2 concrete form}, which implies that the rest are via \eqref{eq: qE def}, \eqref{eq: def UE-1}, \eqref{eq: final def EU}). Therefore for all $U\in\mathbb{R},\varepsilon>0,\mathcal{V}\subset\mathbb{R},\mathcal{Q}\subset(0,1]$ if $V_{U}\left(U\right)\notin\mathcal{V}$ or $q_{U}\left(U\right)\notin\mathcal{Q}$
then for small enough $\varepsilon$ we have recalling the definitions \eqref{Ndef} and above \eqref{eq: loc max quenched} we have
\[
\mathcal{N}_{N}\left(\left[U-\varepsilon,U+\varepsilon\right],\mathcal{V},\mathcal{Q}\right)=-\infty.
\]
This implies that
\[
I_{{\rm TAP}}\left(U,V,q\right)=-\infty\text{\,if }V\ne V_{U}\left(U\right)\text{\,or }q \notin \{0,q_{U}\left(U\right)\}.
\]
Furthermore for any $\varepsilon>0$
\[
\begin{array}{l}
\mathcal{N}_{N}\left(\left[U-\varepsilon,U+\varepsilon\right],\left[V_{U}\left(U-\varepsilon\right),V_{U}\left(U+\varepsilon\right)\right],\left[q_{U}\left(U-\varepsilon\right),q_{U}\left(U+\varepsilon\right)\right]\right)\\
=\mathcal{M}_{N}\left(\left[E_{U}\left(U-\varepsilon\right),E_{U}\left(U+\varepsilon\right)\right]\right)
\end{array}
\]
This proves that for any $U$, if $V=V_{U}\left(U\right)$ and $q=q_{U}\left(U\right)$
we have
\[
I_{{\rm TAP}}\left(U,V,q\right)=I\left(E_{U}\left(U\right)\right).
\]
Furthermore for $\varepsilon$ such that $\varepsilon<q_{\min}$ we
have
\[
\begin{array}{l}
\mathcal{N}_{N}\left(\mathcal{U},\mathcal{V},[0,\varepsilon)\right)=\begin{cases}
1 & \text{ if }\frac{\beta^{2}}{2}\in\mathcal{U},0\in\mathcal{V},\beta\le\frac{\beta^{2}}{2}\\
0 & \text{ otherwise}.
\end{cases}\end{array}
\]
This implies that
\[
I_{{\rm TAP}}\left(U,V,0\right)=\begin{cases}
0 & \text{ if }\frac{\beta^{2}}{2}\in\mathcal{U},0\in\mathcal{V},\beta\le\frac{\beta^{2}}{2}\\
-\infty & \text{ otherwise}.
\end{cases}
\]
Thus the proof of \eqref{eq: TAP comp} is complete.

The first and the third case in \eqref{eq: ITAP verbose} follows from the fact that $E_U(U_{\min})=E_{\min}$, whereby $E_{\min}=E_\infty$ if $\beta \ge \tilde{\beta}_c$ (recall \eqref{eq: Emin def}) and $E_U(U_{\max})=E_0$, which are consequences of the definition \eqref{eq: Umax Umin def} and the fact that $E_U = U_E^{-1}$ (recall \eqref{eq: UE strictly increasing}). The second case follows because $E_U$ is strictly increasing (recall \eqref{eq: EU strictly increasing}) and \eqref{eq: EU of Umin and Umax} so that $E_U(U) \in (E_{\min},E_0)$ if $U\in(U_{\min},U_{\max})$, and that $U_q(q)=U \iff q = q_U(U) \iff q = q_E(E_U(U))$.
\end{proof}

Theorem \ref{thm: TAP comp intro} follows from Theorem \ref{thm:TAP comp}, except for the claim that $U_{\max} \le \beta^2/2$.
This missing part is in fact a consequence of Theorem \ref{thm:fstthrm}.

Finally we rederive \eqref{eq:Itapbeforebetac}, this time from Theorem \ref{thm:TAP comp}, in a way that reinforces the point made in Remark \ref{rem: betac}.
\begin{proof}[Alternative proof \eqref{eq:Itapbeforebetac}] We have that
\begin{equation}\label{eq: betac Emin E0}  
\begin{array}{c}
\beta<\beta_{c}\iff E_{\min}>E_{0},\\
\beta=\beta_{c}\iff E_{\min}=E_{0},\\
\beta>\beta_{c}\iff E_{\min}<E_{0},
\end{array}
\end{equation}
since when $\beta \le \tilde{\beta}_c$ first expression in \eqref{eq: betac appears} is equivalent to $E_{\min}>E_0$ by \eqref{eq: Emin def} and \eqref{eq: rbardef}-\eqref{eq: betacdef} and $E_{\min}$ is non-increasing in $\beta$ (recall \eqref{eq: Emin def}). Therefore 
\begin{equation}\label{eq: EU at least E0}
    \text{when }\beta < \beta_c \text{ we have }\EU(\ETAP) > E_0 \mbox{ for all } \ETAP \ge \Uminbeta,
\end{equation}
so that
$$ \sup_{q\in(0,1],V\in{\mathbb{R}}} \ITAP(U,V,q) \overset{\eqref{eq: TAP comp}}{=} I(E_U(U)) \overset{\eqref{eq: EU at least E0}}{>}I(E_0)=0,$$
which implies that the RHS in fact equals $-\infty$, giving \eqref{eq:Itapbeforebetac}.
\end{proof}

\section{Maximal and minimal TAP energy}\label{sec: umax umin}
In this section we give more concrete formulas for $U_{\min}$ and $U_{\max}$ when $\beta \ge \beta_{c}$.
\begin{prop}\label{prop: umax umin formulas}
It holds for $\beta \ge \tilde{\beta}_c$ that

\begin{equation}\label{eq: Umin form}
U_{\min}=\begin{cases}
\frac{\beta^2}{2}+\frac{1}{2}\log\frac{2}{p}+
\frac{p-2}{8p}\left(4+ \frac{\beta^2}{\tilde{\beta}_{c}^2}\frac{p-2}{p-1}\right)
& \text{\,if }\beta_{c}\le\beta\le\tilde{\beta}_{c},\\
\frac{\beta^{2}}{2}+\frac{1}{2}\log\left(1-q\right)+\frac{q\left(3p\left(1-q\right)+3q-4\right)}{\left(1-q\right)^{2}2p\left(1-p\right)} & \text{\,if }\beta\ge\tilde{\beta}_{c}\\
\text{where }q\text{\,is the unique solution to }\\
\left(1-q\right)q^{\frac{p-2}{2}}=\frac{1}{\beta\sqrt{p\left(p-1\right)}}\text{ in }[\frac{p-2}{p},1).
\end{cases},
\end{equation}
and
\begin{equation}\label{eq: Umax form}
\begin{array}{lll}
U_{\max}&=&\frac{\beta^{2}}{2}+\frac{1}{2}\log\left(1-q\right)+ \frac{2}{p}\frac{q}{1-q}\frac{E_{0}}{E_{\infty}}\bar{r}-\frac{1}{2}\frac{p+\frac{q}{1-q}}{p(p-1)}\frac{q}{1-q}\bar{r}^2\\
& &\text{ where }q\text{\,is the unique solution to}\\
& &\left(1-q\right)q^{\frac{p-2}{2}}=\frac{\bar{r}}{\sqrt{p\left(p-1\right)}}\text{ in }(\frac{p-2}{p},1),
\end{array}
\end{equation}
or alternatively
\begin{equation}\label{eq: Umax as sup}
    U_{\max} = \sup_{q\ge\frac{p-2}{p}: \sqrt{2}\beta_2(q)\le 1} h_{\TAP}(q^{p/2}E_0,q)
\end{equation}
\end{prop}
The following identity will be useful in this section and the next.
\begin{lem}\label{lem: f in terms of w q}
For any $E,q$ we have
\begin{equation}
f\left(E,q\right)=\frac{\beta^{2}}{2}+\frac{1}{2}\log\left(1-q\right)+\frac{E}{E_{\infty}}\frac{2}{p}\frac{q}{1-q}w-\frac{1}{2}\left(p\left(1-q\right)+q\right)\frac{1}{p\left(p-1\right)}\frac{q}{\left(1-q\right)^{2}}w^{2},\label{eq: f in terms of w q}
\end{equation}
where $w=\sqrt{2}\beta_2(q)$.
\end{lem}
\begin{proof}
By \eqref{eq: beta2 concrete form} we have
\begin{equation}\label{eq: q power from w}
    \beta q^{p/2} = \frac{w}{\sqrt{p\left(p-1\right)}}\frac{q}{1-q}\text{ for all }q,\beta.
\end{equation}
Thus for the first term of $f\left(E,q\right)$ in \eqref{eq: f def} we have that
\[
\beta Eq^{p/2}
\overset{E_{\infty}=2\sqrt{\frac{p-1}{p}}}{=} \frac{E}{E_\infty} \frac{2}{p} \frac{q}{1-q}w
\]
for all $q$. By \eqref{eq: Onsager concrete} the last term of $f\left(E,q\right)$ equals
\[
\begin{array}{lcl}
\frac{\beta^{2}}{2}\left(1-p\left(1-q\right)q^{p-1}-q^{p}\right)& = & \frac{\beta^{2}}{2}-\frac{\beta^{2}}{2}q^{p-1}\left(p\left(1-q\right)+q\right)\\
 & \overset{\eqref{eq: beta2 concrete form}}{=} & \frac{\beta^{2}}{2}-\frac{1}{2}\left(\sqrt{2}\beta_{2}\left(q\right)\right)^{2}\frac{1}{p\left(p-1\right)}\frac{q}{\left(1-q\right)^{2}}\left(p\left(1-q\right)+q\right),
\end{array}
\]
for all $q$. Thus \eqref{eq: f in terms of w q} follows.
\end{proof}
\begin{proof}[Proof of Proposition \ref{prop: umax umin formulas}]
Recall from \eqref{eq: Umax Umin def} and \eqref{eq: Uminbeta}-\eqref{eq: def UE-1} that $U_{\min} = f(E_{\min},q_{\min})$ and $U_{\max} = f(E_0,q_E(E_0))$. The latter with \eqref{eq: qE glob max} implies \eqref{eq: Umax as sup}.

From \eqref{eq: qE as sol conc} we have that $q_{\min}=q_E(E_{\min})$ satisfies $w=\sqrt{2}\beta_2(q_{\min})=\frac{\beta}{\tilde{\beta}_c}$ for $\beta\le \tilde{\beta}_c$ and $w=\sqrt{2}\beta_2(q_{\min})=1$ for $\beta \ge \tilde{\beta}_c$.

When $\beta \in [\beta_c,\tilde{\beta}_c]$ we have $E_{\min}=\frac{E_\infty}{2}\left(\frac{\tilde{\beta}_{c}}{\beta}+\frac{\beta}{\tilde{\beta}_{c}}\right)$ by \eqref{eq: Emin def}.
Thus in this case using Lemma \ref{lem: f in terms of w q} we get that
$$U_{\min} = \frac{\beta^2}{2}+\frac{1}{2}\log\frac{2}{p}+\left(\frac{\beta}{\tilde{\beta}_{c}}+\frac{\tilde{\beta}_{c}}{\beta}\right)\frac{1}{p}\frac{q}{1-q}\frac{\beta}{\tilde{\beta}_{c}}-\frac{1}{2}\left(p\left(1-q\right)+q\right)\frac{1}{p\left(p-1\right)}\frac{q}{\left(1-q\right)^{2}}\frac{\beta^2}{\tilde{\beta}_{c}^2},$$
where $q=q_{\min}$. Since by \eqref{eq: qmin def} we have $q_{\min}=\frac{p-2}{p}$ this simplifies to the first line of the LHS of \eqref{eq: Umin form}.

When $\beta \ge \tilde{\beta}_c$ we have $E_{\min} = E_\infty$. By \eqref{eq: qmin def} and \eqref{eq: beta2 concrete form} we get that $q=q_{\min}=q_E(E_{\infty})$ is the unique solution to the equation in the bottom line of the LHS of \eqref{eq: Umin form}. Thus by Lemma \ref{lem: f in terms of w q} we get that
$$ U_{\min} = \frac{\beta^{2}}{2}+\frac{1}{2}\log\left(1-q\right)+\frac{2}{p}\frac{q}{1-q}-\frac{1}{2}\left(p\left(1-q\right)+q\right)\frac{1}{p\left(p-1\right)}\frac{q}{\left(1-q\right)^{2}},$$
where $q=q_{\min}$, which simplifies to the second line of \eqref{eq: Umin form}.

Moving to $U_{\max}$, we have that $q_E(E_0)$ solves the equation in \eqref{eq: Umax form} by \eqref{eq: qE as sol conc}. By Lemma \ref{lem: f in terms of w q} we get that

\[
U_{\max} = \frac{\beta^2}{2} + \frac{1}{2}\log(1-q) +
\frac{E_{0}}{E_{\infty}}\frac{2}{p}\frac{q}{1-q}\bar{r}-\frac{1}{2}\left(p\left(1-q\right)+q\right)\frac{1}{p\left(p-1\right)}\frac{q}{\left(1-q\right)^{2}}\bar{r}^{2},
\]
which simplifies to the first line on the LHS of \eqref{eq: Umax form}.
\end{proof}

\section{Solution of the optimization problem}\label{optimize}
This section is devoted to the analysis of the TAP free energy
\beq \fTAP(\beta) = \sup_{\ETAP \in\mathbb{R},V\in\mathbb{R},q\in D_\beta} \{\ETAP + \ITAP(\ETAP,V,q)\} \eeq
and the proofs of Theorems \ref{thm:sndthrm}-\ref{thm:fstthrm}.   
First we rewrite the optimization over $\ETAP$ as an optimization over $E$ and $q$ as follows:
\begin{lem}\label{lem:reformu}
For any $\beta\geq 0$ and $U\ge U_{\min},q\ne0$
\begin{equation}\label{eq: TOT TAP FE from f and I}
U+I_{{\rm TAP}}\left(U,V,q\right)=\begin{cases}
f\left(E,q\right)+I\left(E\right) & \text{\,if }q=q_{E}\left(E\right),V=q^{p/2}E\text{\,for }E=E_U(U),\\
-\infty & \text{ otherwise}.
\end{cases}
\end{equation}
where $f(E,q)$ is defined in \eqref{eq: f def}, $q_E$ in \eqref{eq: qE glob max} and $E_U$ in \eqref{eq: final def EU}. Also
\begin{equation}\label{eq: reformu}
\begin{array}{ccc}
    \displaystyle{\sup_{U\in\mathbb{R},V\in\mathbb{R},q \in D_\beta\setminus\{0\}}}\left\{ U+\ITAP\left(U,V,q\right)\right\} 
    &=&\displaystyle{\sup_{E\in [\Emin ,E_0]}}\left\{ f\left(E,q_E(E)\right)+I\left(E\right)\right\}\\
    &=&\displaystyle{\sup_{E\in [\Emin ,E_0],q\in\left[q_{\min},1\right)}}\left\{ f\left(E,q\right)+I\left(E\right)\right\}.
\end{array}    
\end{equation}
\end{lem}
\begin{proof}
For any $U\ge \Uminbeta$ we have $U=U_{E}\left(E_{U}\left(U\right)\right)$ since
 $E_U$ is the inverse of $U_E$, and \eqref{eq: TOT TAP FE from f and I} then follows from the definition \eqref{eq: def UE-1} of $U_E$ and \eqref{eq: TAP comp}. From \eqref{eq: TOT TAP FE from f and I} the first equality of \eqref{eq: reformu} follows since the range of $E_U$ is $[E_{\min},\infty)$ and $I(E) = -\infty$ for $E>E_0$. The second inequality follows from \eqref{eq: qE glob max} and \eqref{eq: Dbeta is interval in dyn low temp}), since the range of $q_E$ is $[q_{\min},1)$.
\end{proof}
Recall
\begin{equation}\label{eq: IAnn def}
I_{\Ann}(E) = g\left(\frac{E}{E_\infty}\right),
\end{equation}
for $g$ from \eqref{eq:g} denote the annealed rate function of local maxima. In what follows we will will compute
\begin{equation}\label{eq: annealed}
\sup_{E\in [\Emin,\infty)}\left\{ f\left(E,q_E(E)\right)+I_{\Ann}\left(E\right)\right\} \overset{\eqref{eq: qE glob max}}{=}
\sup_{E\in [\Emin,\infty),q\in[q_{\min},1)}\left\{ f\left(E,q\right)+I_{\Ann}\left(E\right)\right\} 
\end{equation}
Using that $I=I_{\Ann}$ on $[E_\infty,E_0]$ and $I= -\infty$ on $(E_0,\infty)$ we will be able to derive from this the value of \eqref{eq: reformu}. Since $q_E$ maximizes $q\to f(E,q)$ (see \eqref{eq: qE glob max}) the useful identity
\begin{equation}\label{eq: useful ident}
\frac{d}{dE}\left\{ f\left(E,q_{E}\left(E\right)\right)+I_{\Ann}\left(E\right)\right\} = \partial_{E}f\left(E,q_{E}\left(E\right)\right)+I_{\Ann}'\left(E\right)\text{ for }E\ge E_{\min}
\end{equation}
holds.

To compute \eqref{eq: annealed} we will consider the critical point equations
\begin{equation}\label{eq: crit point eq ann TAP}
    \begin{array}{ccc}
         \partial_E \left( f(E,q) + I_{\Ann}(E) \right) & = & 0, \\
         \partial_q \left( f(E,q) + I_{\Ann}(E) \right) &= & 0.
    \end{array}
\end{equation}
The second of these equations is nothing but the equation $\partial_q f(E,q)=0$ whose solutions are studied in Section \ref{sec:det}. Indeed Lemma \ref{lem: critequiv} implies that
\begin{equation}\label{eq: E tilde eq}
(\tilde{E},\tilde{q})\text{ satisfies }\partial_{q}f\left(E,q\right)=0\text{\,iff }\frac{\tilde{E}}{E_\infty}=\frac{1}{2}\left(\sqrt{2}\beta_{2}\left(\tilde{q}\right)+\frac{1}{\sqrt{2}\beta_{2}\left(\tilde{q}\right)}\right).
\end{equation}
The following identity will be useful to study the first of equation in \eqref{eq: crit point eq ann TAP}.
\begin{lem}\label{lem: IAnn deriv ident}
If $E=\frac{E_{\infty}}{2}\left(v+\frac{1}{v}\right)$ for $v\in\left(0,1\right)$
then
\begin{equation}\label{eq: IAnn deriv ident}
I_{\Ann}'\left(E\right)=\frac{v}{\sqrt{p\left(p-1\right)}}-\sqrt{\frac{p-1}{p}}\frac{1}{v}.
\end{equation}
\end{lem}
\begin{proof}
By \eqref{eq: IAnn def} we have 
$I_{\Ann}'\left(E\right) = \frac{1}{E_\infty} g'\left(\frac{E}{E_\infty}\right).$
Note that
$$ \Omega'(\eta) = 2(\eta - \sqrt{\eta^2-1}),$$
(either by direct computation from the RHS of \eqref{eq: omega}, or since $\Omega'(\eta)$ it is the Stieltjes transform of the semi-circle law, as can be seen from taking the derivative of the integral in \eqref{eq: omega}) so that from \eqref{eq:g} we have
$$ g'(\eta) = -4\frac{p-1}{p} \eta + 2\left(\eta - \sqrt{\eta^2-1}\right).$$
Note that
\begin{equation}\label{eq: eta v ident}
\text{if }\eta=\frac{1}{2}\left(v+\frac{1}{v}\right)\text{ for some }v\in\left(0,1\right)\text{\,then }\sqrt{\eta^{2}-1}=\frac{1}{2}\left( \frac{1}{v}-v \right)
\end{equation}
so that
$$ \frac{g'(\eta)}{E_\infty} = \frac{1}{E_\infty} \left( -2\frac{p-1}{p}\left( v+ \frac{1}{v}\right) + 2v \right) = \frac{1}{E_\infty} \left( \frac{2}{p}v-2\frac{p-1}{p}\frac{1}{v} \right)
\overset{\eqref{eq: Einf def}}{=}
\frac{v}{\sqrt{p\left(p-1\right)}}-\sqrt{\frac{p-1}{p}}\frac{1}{v}.$$
\end{proof}
The identity implies the following about solutions to the second equation of \eqref{eq: crit point eq ann TAP}.
\begin{lem}\label{lem: E deriv at q crit point}
If $\left(E,q\right)$ satisfies $\partial_{q}f\left(E,q\right)=0$
then at $\left(E,q\right)$ we have

\begin{equation}\label{eq: E deriv at q crit point}
\partial_{E}\left(
f\left(E,q\right)+I_{\Ann}\left(E\right)
\right)
= \frac{1}{\sqrt{p}}\left( \frac{w}{\sqrt{p-1}}\frac{1}{1-q}-\frac{\sqrt{p-1}}{w} \right),
\end{equation}
where $w=\sqrt{2}\beta_2(q)$.
\end{lem}

\begin{proof}
By the definition \eqref{eq: f def} of $f$ we have
$$\partial_{E}\left(f\left(E,q\right)+I_{\Ann}\left(E\right)\right)=\beta q^{p/2}+I_{\Ann}'(E)$$ for all $E,q$. Thus by the previous lemma and \eqref{eq: E tilde eq}, if
 $(E,q)$ satisfies $\partial_q f(E,q)=0$ then
$$\partial_{E}\left(f\left(E,q\right)+I_{\Ann}\left(E\right)\right) = \beta q^{p/2}+\frac{w}{\sqrt{p\left(p-1\right)}}-\sqrt{\frac{p-1}{p}}\frac{1}{w},$$
where $w=\sqrt{2}\beta_2(q)$. Rewriting the first term using \eqref{eq: q power from w} and simplifying \eqref{eq: E deriv at q crit point} follows.
\end{proof}

We now solve the critical point equations \eqref{eq: crit point eq ann TAP}.
\begin{lem}\label{lem: crit eq sol}
If $\beta > \beta_d$ the critical point equations \eqref{eq: crit point eq ann TAP} have two solutions $(\tilde{E},\tilde{q})$ whereby $\tilde{q}$ is any of the two solutions of
\beq \label{eq: qtildedef}(1-q)q^{p-2}=\frac{1}{p\beta^2},
\eeq
or of the equivalent equation
\beq \label{eq: qtildedef equiv}
2 \beta_2(q)^2 = (p-1)(1-q),
\eeq
and
\beq \label{eq: Etildedef}\tilde{E} = \frac{E_\infty}{2}\left(\frac{1}{\sqrt{2}\beta_2(\tilde{q})} +\sqrt{2} \beta_2(\tilde{q})\right).
\eeq
Exactly one of the solutions lies in $[E_\infty,\infty)\times [ q_{\min},1)$, and comes from the unique solution to \eqref{eq: qtildedef} in $[ \frac{p-2}{p-1}, 1]$. Both $\tilde{q}$ and $\tilde{E}$ are strictly increasing as functions of $\beta$.

If $\beta < \beta_d$ then at any $(\tilde{E},\tilde{q}) \in [E_\infty,\infty)\times[0,1]$ such that \eqref{eq: Etildedef} holds satsifies
\begin{equation}\label{eq: negative}
    \frac{d}{dE} \left( f(E,q) + I_{\Ann}(E) \right) |_{(E,q)=(\tilde{E},\tilde{q})}< 0,
\end{equation}
and so the critical point equations \eqref{eq: crit point eq ann TAP} have no solution in $[E_\infty,\infty)\times[0,1].$
\end{lem}
\begin{rem}
Recall from \eqref{eq: qmin at betad} that $q_{\min} > \frac{p-2}{p-1}$ when $\beta > \beta_d$, so then $[q_{\min},1) \subset [ \frac{p-2}{p-1}, 1]$.
\end{rem}
\begin{proof}
By \eqref{eq: E tilde eq} and the previous lemma a $(\tilde{E},\tilde{q})$ which satisfies \eqref{eq: Etildedef} also satisfies
$$
\partial_{E}\left(
f\left(E,q\right)+I_{\Ann}\left(E\right)
\right) |_{(E,q)=(\tilde{E},\tilde{q})}
= \frac{(p-1)(1-\tilde{q})}{w(1-\tilde{q})\sqrt{p(p-1)}} \left( \frac{w^2}{(p-1)(1-\tilde{q})}-1 \right).
$$
Now by \eqref{eq: beta2 concrete form} we have
$$ \frac{w^2}{(p-1)(1-q)} = p\beta^2 (1-q) q^{\frac{p-2}{2}},$$
and
\begin{equation}\label{eq: max over q}
q\to p \beta^{2}\left(1-q\right)q^{p-2}\text{\,is maximized at }\frac{p-2}{p-1}\text{ where it takes the value }\left(\frac{\beta}{\beta_{d}}\right)^{2},
\end{equation}
(recall \eqref{eq: beta d def}) and strictly decreasing thereafter. This implies \eqref{eq: negative}. It also implies that indeed \eqref{eq: qtildedef}, \eqref{eq: qtildedef equiv} has two solutions if $\beta>\beta_d$, one in $(0,\frac{p-2}{p-1})$ and one in $(\frac{p-2}{p-1},1)$, that each of these correspond to a solution of \eqref{eq: crit point eq ann TAP}, and that there are no other solutions of \eqref{eq: crit point eq ann TAP}. Since $p\beta^2(1-q)q^{p-2}$ is increasing in $\beta$ this also shows that $\tilde{q}$ is increasing in $\beta$, and since $\sqrt{2} \beta_2(q)$ is increasing in $\beta$ (recall \eqref{eq: beta2 concrete form}) we get from \eqref{eq: Etildedef} that $\tilde{E}$ is also increasing in $\beta$.
It only remains to show the larger solution to \eqref{eq: qtildedef}, \eqref{eq: qtildedef equiv} in fact lies in $[q_{\min},1)$. It can be checked that $\sqrt{p(p-1)} (1-q) q^{\frac{p-2}{2}} \le p\beta^2(1-q)q^{p-2}$ for $q\ge\frac{p-2}{p-1},\beta > \beta_d$, and from this \eqref{eq: beta2 concrete form}, \eqref{eq: qmin def} and \eqref{eq: qtildedef} that claim follows.
\end{proof}
The following identity will be useful to compute $f(E,q)+I(E)$.

\begin{lem}
If $\eta=\frac{1}{2}\left(v+\frac{1}{v}\right)$ for some $v\in\left(0,1\right)$ then
\begin{equation}\label{eq: Iann ident}
I_{\Ann}\left(E_{\infty}\eta\right) = g\left(\eta\right)
= \frac{1}{2}\left(v+\frac{1}{v}\right)\left(v-\frac{p-1}{p}\left(v+\frac{1}{v}\right)\right)-\log\frac{v}{\sqrt{p-1}}.
\end{equation}
\end{lem}
\begin{proof}
We have
\[
\Omega\left(\eta\right)\overset{\eqref{eq: omega},\eqref{eq: eta v ident}}{=}\frac{1}{2}v^{2}-\log v,
\]
so that by \eqref{eq:g} and \eqref{eq: eta v ident}
\begin{equation}\label{eq: Iann ident in proof}
I_{\text{Ann}}\left(E_{\infty}\eta\right)  =  g\left(\eta\right)
  =  \frac{1}{2}-\frac{p-1}{p}\frac{1}{2}\left(v+\frac{1}{v}\right)^{2}+\frac{1}{2}v^{2}-\log\frac{v}{\sqrt{p-1}},
\end{equation}
giving \eqref{eq: Iann ident}.
\end{proof}
Next we compute the value of $f(E,q)+I(E)$ when $\beta > \beta_d$ at the unique solution from Lemma \ref{lem: crit eq sol} with $q\ge q_{\min}$.
\begin{lem}\label{lem: val at max}
Assume $\beta>\beta_{d}$. The unique solution $(\tilde{E},\tilde{q})$ of \eqref{eq: crit point eq ann TAP} with $\tilde{q}\ge q_{\min}$ from Lemma \ref{lem: crit eq sol} satisfies
\begin{equation}\label{eq: val at max}
 f\left(\tilde{E},\tilde{q}\right)+I_{\Ann}\left(\tilde{E}\right) =\frac{\beta^{2}}{2},
\end{equation}
and
\begin{equation}\label{eq: entropy at max}
I_{\Ann}(\tilde{E})=-\frac{\tilde{q}}{2}-\frac{1}{2p}\frac{\tilde{q}^{2}}{1-\tilde{q}}-\frac{1}{2}\log(1-\tilde{q}).
\end{equation}
\end{lem}
\begin{proof}
We evaluate \eqref{eq: f in terms of w q} and \eqref{eq: Iann ident} at $(\tilde{E},\tilde{q})$.
By \eqref{eq: qtildedef equiv} we have
\begin{equation}\label{eq: q from w}
\frac{1}{1-\tilde{q}}=\frac{p-1}{w^{2}}
\end{equation}
as well as
\begin{equation}\label{eq: pqw ident}
p\left(1-\tilde{q}\right)+\tilde{q}=\left(p-1\right)\left(1-\tilde{q}\right)+1=w^{2}+1.
\end{equation}
With these identities the last term of (\ref{eq: f in terms of w q}) becomes $\frac{1}{2}\frac{p-1}{p}\left(w+\frac{1}{w}\right)\frac{\tilde{q}}{w}$.
Using also $\frac{\tilde{E}}{E_{\infty}}=\frac{1}{2}\left(w+\frac{1}{w}\right)$
the third term of (\ref{eq: f in terms of w q}) becomes $\frac{p-1}{p}\left(w+\frac{1}{w}\right)\frac{\tilde{q}}{w}$,
so we get that 
\[
f\left(\tilde{E},\tilde{q}\right)=\frac{1}{2}\beta^{2}+\frac{1}{2}\log\left(1-\tilde{q}\right)+\frac{1}{2}\frac{p-1}{p}\left(w+\frac{1}{w}\right)\frac{\tilde{q}}{w}.
\]
Also setting $v=w$ the last term of \eqref{eq: Iann ident} becomes $-\frac{1}{2}\log(1-\tilde{q})$ and the last factor of the first term becomes $-\frac{wq}{p(1-q)}$. Thus
\begin{equation}\label{eq: Iann ident 2}
I_{\Ann}\left(\tilde{E}\right)= -\frac{1}{2}\left(w+\frac{1}{w}\right)\frac{w\tilde{q}}{p(1-\tilde{q})}-\frac{1}{2}\log\left(1-\tilde{q}\right).
\end{equation}
Adding these we get
\[
f\left(\tilde{E},\tilde{q}\right)+I_{\Ann}\left(\tilde{E}\right)=\frac{\beta^{2}}{2}+\frac{1}{2}\left(w+\frac{1}{w}\right)\frac{\tilde{q}}{p}\left(\frac{p-1}{w}-\frac{w}{p(1-\tilde{q})}\right),
\]
and the last factor of the last term on the RHS is
\[
\frac{p-1}{w}-\frac{w}{p(1-\tilde{q})}
=
\frac{1}{w}\left(p-1-\frac{w^2}{p(1-\tilde{q})}\right) \overset{\eqref{eq: q from w}}{=} 0,
\]
proving \eqref{eq: val at max}.
Also from \eqref{eq: Iann ident 2} we get
\[
I_{\Ann}\left(\tilde{E}\right)=-\frac{1}{2}\left(w^2+1\right)\frac{\tilde{q}}{p(1-\tilde{q})}-\frac{1}{2}\log(1-\tilde{q}),
\]
which with \eqref{eq: pqw ident} implies \eqref{eq: entropy at max}.
 \end{proof}
Next we confirm that when $\beta>\beta_d$ then $(\tilde{E},\tilde{q})$ is in fact the unique point that achieves the supremum in \eqref{eq: annealed}. Note that 
\begin{equation}\label{eq: Emin is Einf over betad}
    E_{\min}=E_\infty\text{ when }\beta>\beta_d,
\end{equation}
(recall \eqref{eq: Emin def} and \eqref{eq: order of betas}).
\begin{lem}[Annealed TAP variational formula equals annealed free energy for $\beta>\beta_d$]\label{lem:solveopti}
If $\beta>\beta_{d}$ then
\begin{equation}\label{eq: sup annealed}
\sup_{E\in [E_{\infty},\infty),q\in[q_{\min},1)}\left\{ f\left(E,q\right)+I_{\Ann}\left(E\right)\right\} =\frac{\beta^{2}}{2},
\end{equation}
and the supremum is achieved at a unique point $(\tilde{E},\tilde{q})$. This point is the unique solution from Lemma \ref{lem: crit eq sol} that lies in $(E_{\infty},\infty)\times(q_{\min},1)$.
\end{lem}
\begin{proof}
We must check that the supremum is achieved in $(E_{\infty},\infty)\times(q_{\min},1)$. The claims then follows by Lemmas \ref{lem: crit eq sol} and \ref{lem: val at max}.

Note that $f(E,q) + I_{\Ann}(E) \to -\infty$ as $q\to1$ or $E\to \infty$ (since $I_{\Ann}(E)$ goes to $-\infty$ quadratically and $f(E,q)$ to $\infty$ only linearly as $E\to \infty$). This implies that the supremum of \eqref{eq: sup annealed} must be achieved at a point in $[E_{\infty}, \infty) \times [q_{\min},1)$, since $f(E,q) + I_{\Ann}(E)$ is continuous on this set.

Considering the border $E=E_\infty$ note that using \eqref{eq: IAnn deriv ident} with $v=1$ it holds for all $q\ge q_{\min}\ge (p-2)/(p-1)$ (recall \eqref{eq: qmin at betad}) that
\beq \label{eq: deriv at Einf}
\frac{\partial}{\partial E}\left(f(E,q)+I_{\Ann}(E)\right)|_{E=E_{\infty}}  =  \beta q^{\frac{p}{2}}-\frac{p-2}{\sqrt{p(p-1)}} >  \beta_{d}\left(\frac{p-2}{p-1}\right)^{\frac{p}{2}}-\frac{p-2}{\sqrt{p(p-1)}} \overset{\eqref{eq: beta d def}}{=}0,
\eeq
showing that
the supremum of \eqref{eq: sup annealed} is achieved at a point in $(E_\infty,\infty)\times[q_{\min},1)$.
Finally, by Lemma \ref{lem: E g Emin loc max}, for all $E>E_\infty$ the function $q\to f(E,q)$ has only one critical point in  $(q_{\min},1)$ which is a local maximum, which implies that the supremum is in fact achieved at a point in $(E_\infty,\infty)\times(q_{\min},1)$.
Thus the maximizer must satisfy \eqref{eq: crit point eq ann TAP}, so it is the unique solution $(\tilde{E},\tilde{q})$ from Lemma \ref{lem: crit eq sol}. By Lemma \ref{lem: val at max} the equality follows.
\end{proof}

Next we turn our attention to the value of the RHS of \eqref{eq: annealed} when $\beta < \beta_d$.
\begin{lem}\label{lem: dyn HT maximizer}
For $\beta\in [0, \beta_d)$ it holds that
$f(E,\qE(E)) + I_{\Ann}(E)$ is decreasing in $E$ on $[E_{\min},\infty)$.
\end{lem}
\begin{proof}
For $E\ge E_{\min}$ note that 
\[
\frac{d}{dE}\left\{ f\left(E,q_{E}\left(E\right)\right)+I_{\Ann}\left(E\right)\right\} \overset{\eqref{eq: useful ident}}{=} \partial_{E}f\left(E,q_{E}\left(E\right)\right)+I_{\Ann}'\left(E\right) \overset{\eqref{eq: negative}}{<} 0.
\]
\end{proof}

The next lemma shows that for  $\beta < \beta_d$ (that is, in static and dynamic high temperature) all non-zero \TP\ solution have a TAP energy lower than the TAP energy of $m=0$.

\begin{lem}[Annealed TAP variational formula for $\beta\le \beta_d$]\label{lem:solveopti2}
For $\beta\in [0, \beta_d]$
\beq\label{eq:solvopti2}
f(E_{\min}, q_{\min}) + I_{\Ann}(E_{\min})
\le \frac{\beta^2}{2},
\eeq
with equality only if $\beta=\beta_d$.
\end{lem}
\begin{proof}
We evaluate the LHS of \eqref{eq:solvopti2} considering the cases $\beta\in[0,\tilde{\beta}_{c}]$
and $\beta\in[\tilde{\beta}_{c},\beta_{d}]$ separately. In the first case $\beta\in[0,\tilde{\beta}_{c}]$
we have $q_{\min}=\frac{p-2}{p}$ and $\frac{E_{\min}}{E_\infty}=\frac{1}{2}\left(w+\frac{1}{w}\right)$ where $w=\frac{\beta}{\tilde{\beta}_c}\le1$ (recall \eqref{eq: w def}, \eqref{eq: Emin def}, \eqref{eq: qmin def}).
Plugging these value for $q$ into \eqref{eq: f in terms of w q} we
get
\[
f\left(E_{\min},q_{\min}\right)=\frac{1}{2}\beta^{2}+\frac{1}{2}\log\frac{2}{p}+\frac{1}{2}\left(w+\frac{1}{w}\right)w\frac{p-2}{\textbf{}p}-\frac{1}{2}w^{2}\frac{1}{4}\frac{3p-2}{p}\frac{p-2}{p-1}
\]
We also have from \eqref{eq: Iann ident}
\[
I_{\text{Ann}}\left(E_{\min}\right)=\frac{1}{2}\left(w+\frac{1}{w}\right)\left(w-\frac{p-1}{p}\left(w+\frac{1}{w}\right)\right)-\log\frac{w}{\sqrt{p-1}},
\]
so that letting $x=\frac{2\left(p-1\right)}{pw^{2}}$ we get
\[
\begin{array}{l}
f\left(E_{\min},q_{\min}\right)+I_{\text{Ann}}\left(E_{\min}\right)\\
=\frac{\beta^{2}}{2}+\frac{1}{2}\left(w+\frac{1}{w}\right)w\frac{p-2}{p}+\frac{1}{2}\log x-\frac{1}{2}\frac{3p-2}{p}\frac{1}{4}w^{2}\frac{p-2}{p-1}+\frac{1}{2}\left(w+\frac{1}{w}\right)\left(w-\frac{p-1}{p}\left(w+\frac{1}{w}\right)\right)\\
=\frac{\beta^{2}}{2}+\frac{1}{2}\frac{p-1}{p}\left(w+\frac{1}{w}\right)\left(w-\frac{1}{w}\right)+\frac{1}{2}\log x-\frac{1}{2}\frac{3p-2}{p}\frac{1}{4}w^{2}\frac{p-2}{p-1}\\
=\frac{\beta^{2}}{2}+\frac{1}{2}\left(\frac{p-1}{p}-\frac{3p-2}{p}\frac{1}{4}\frac{p-2}{p-1}\right)w^{2}-\frac{x}{4}+\frac{1}{2}\log x\\
=\frac{\beta^{2}}{2}+\frac{1}{4x}-\frac{x}{4}+\frac{1}{2}\log x.
\end{array}
\]
Note that $\frac{1}{4z}-\frac{z}{4}+\frac{1}{2}\log z<0$ for all $z>1$. Furthermore since $w\le1$ we have $x\ge2\frac{p-1}{p}>1$. This proves \eqref{eq:solvopti2} in this case.

Next we consider the case $\beta\in[\tilde{\beta}_{c},\beta_{d}]$.
In this case $E_{\min}=E_{\infty}$ (recall \eqref{eq: w def} and \eqref{eq: Emin def}), and from \eqref{eq: Iann ident}
\begin{equation}\label{eq: Iann at Einf}
I_{\text{Ann}}\left(E_{\infty}\right)=-\frac{p-2}{p}+\frac{1}{2}\log\left(p-1\right).
\end{equation}
Also from \eqref{eq: f in terms of w q}, recalling that $\sqrt{2}\beta(q_{\min})=1$ by \eqref{eq: qmin def}, we get that with $r=\frac{q_{\min}}{1-q_{\min}}$ or equivalently $1-q_{\min}=\frac{1}{1+r}$ that
\[
f\left(E_{\min},q_{\min}\right)=\frac{\beta^{2}}{2}-\frac{1}{2}\log\left(1+r\right)+\frac{2}{p}r-\frac{1}{2}r\frac{1}{p-1}-\frac{1}{2}r^2\frac{1}{p\left(p-1\right)}.
\]
Now \eqref{eq: qmin at betad} implies that $q_{\min}\ge\frac{p-2}{p-1}$ for $\beta\le\beta_d$ with equality only if $\beta=\beta_d$. Since $q_{\min}\ge\frac{p-2}{p-1}$ corresponds to $r\ge p-2$ with equality only if $r=p-2$ the claim of the lemma in the case $\beta \in [\tilde{\beta}_c,\beta_d]$ follows once we have shown the bound
$$
-\frac{1}{2}\log\left(1+r\right)+\frac{2}{p}r-\frac{1}{2}r\frac{1}{p-1}-\frac{1}{2}r^2\frac{1}{p\left(p-1\right)} - \frac{p-2}{p} + \frac{1}{2}\log(p-1) \le 0 \text{ for } r\ge p-2,$$
(cf. \eqref{eq: Iann at Einf}) with equality only if $r=p-2$. It is easy to verify that the LHS is zero when $r=p-2$. Also the derivative of the LHS is zero when $r=p-2$, and the second derivative is negative for $r\ge p-2$. This gives the bound and finishes the proof of the lemma.
\end{proof}

We proceed to derive Theorems \ref{thm:sndthrm}, \ref{thm:dynLTstatHT} and \ref{thm:fstthrm}. 
We do so by reconstructing the maximizer of $f(E,q)+I(E)$ from the maximizer of $f(E,q)+I_{\Ann}(E)$ by considering when the maximizer of the latter hits $E=E_0$.

\begin{proof}[Proof of Theorem \ref{thm:sndthrm}]
If $\beta<\beta_d$ then by \eqref{eq:varprinc} and \eqref{eq: TAP comp}
\beq \fTAP(\beta) =\max\left(\frac{\beta^2}{2}, \sup_{\ETAP \in \R,V\in \mathbb{R},q\in D_\beta\setminus\{0\}} \{ \ETAP+\ITAP(\ETAP,V,q)\}\right),\eeq
where the contribution $\frac{\beta^2}{2}$ is due to $q=0,V=0$ (i.e. $m=0$). By Lemmas \ref{lem:reformu}, \ref{lem: dyn HT maximizer} and \ref{lem:solveopti2} we see that the other quantity in the max is strictly smaller. Furthermore the other quantity is an upper bound for $U_{\max}$, proving \eqref{eq: FE very high temp 2}. This also shows that \eqref{eq: def E TAP star} is well-defined shows all the claims in \eqref{eq: qstar Ustar entropy high temp}, by the second case in \eqref{eq: TAP comp}.\end{proof}

To prove Theorem \ref{thm:dynLTstatHT} we need the following characterization of $\beta_s$.
\begin{lem}
When $\beta>\beta_d$ the energy $\tilde{E}$ from Lemma \ref{lem: crit eq sol} satisfies
\begin{equation}\label{eq: E0 betas equiv}
\begin{array}{c}
\beta<\beta_{s}\iff\tilde{E}<E_{0},\\
\beta=\beta_{s}\iff\tilde{E}=E_{0},\\
\beta>\beta_{s}\iff\tilde{E}>E_{0}
\end{array}
\end{equation}
\end{lem}
\begin{proof}
Letting $w=\sqrt{2}\beta_2(\tilde{q})$ note that
\[
\begin{array}{ll}
\tilde{E}<E_{0} \overset{\eqref{eq: Etildedef}}{\iff} \frac{1}{2}\left(\frac{1}{w}+w\right)<\frac{E_{0}}{E_\infty}
\overset{\eqref{eq: rpmdef}, \eqref{eq: rpmdef as sols}}{\iff} w > r_{-}\left(\frac{E_{0}}{E_\infty}\right)\\
\overset{\eqref{eq: rbardef}, \eqref{eq: qtildedef equiv}}{\iff} \sqrt{1-\tilde{q}}\sqrt{p-1} > \bar{r} 
\iff
\tilde{q}<1-\frac{1}{p-1}\bar{r}^{2}.
\end{array}.
\]
Now recalling that $\tilde{q}$ is increasing in $\beta$ by Lemma \ref{lem: crit eq sol} the claim follows by noting that at $\beta=\beta_s$ we have $\tilde{q}=\hat{q}=1-\frac{1}{p-1}\bar{r}^2$, which can be verified by observing that 
$$ p \beta_s^2(1-\hat{q}) \hat{q}^{p-2} \overset{\eqref{betasdef}}{=}1.$$
\end{proof}

\begin{proof}[Proof of Theorem \ref{thm:dynLTstatHT}]
Let $\beta \in (\beta_d,\beta_s)$. Then by Lemma \ref{lem:ourplefkaprofs} 1) we have $0\not\in D_\beta$ so that from \eqref{eq:varprinc}, \eqref{eq: reformu} and \eqref{eq: Emin is Einf over betad}
\beq \label{eq: sup over E q}
\fTAP(\beta) =\sup_{E \in [E_{\infty},E_0],q\in[q_{min},1)} \left\{ f(E,q) + I(E) \right\}.
\eeq
By Lemma \ref{lem:solveopti} the supremum is $f(\tilde{E},\tilde{q})=\beta^2/2$, attained uniquely at $(\tilde{E},\tilde{q})$ with $\tilde{E}>E_{\infty}$, as long as the maximizer $\tilde{E}$ from Lemma \ref{lem:solveopti} satisfies $\tilde{E}\in(E_\infty,E_0]$, since then the global maximizer is within the region where $I_{\Ann}= I$. This with the previous lemma proves \eqref{eq: FE dyn high temp}. By \eqref{eq: qstar dyn low temp} we get that $\tilde{q}$ is the unique solution to \eqref{eq: qstar dyn low temp}. Since the supremum in \eqref{eq: sup over E q} is uniquely attained, we get from \eqref{eq: TOT TAP FE from f and I} and the fact that $E_U$ is a bijection from $[U_{\min},\infty)$ to $[E_{\min},\infty)$ that the supremum in \eqref{eq: def E TAP star} is uniquely attained, so that $U_*,V_*,q_*$ are well-defined and $U_*=f(\tilde{E},q_*),V_*=q_{*}^{p/2}\tilde{E},q_*=\tilde{q}$, which implies \eqref{eq: Estar dyn LT stat HT} and \eqref{eq: Ustar dyn LT stat HT} (recall \eqref{eq: Etildedef}). That $\tilde{E} \in (E_{\infty}, E_0)$ together with \eqref{eq: Umax Umin def} implies \eqref{eq:EstarnotGS}. Lastly the identity \eqref{eq: entropy dyn low temp} follows from \eqref{eq: entropy at max} and the inequality since $\tilde{E} \in (E_{\infty},E_0)$  and \eqref{eq: I props}.
\end{proof}

\begin{proof}[Proof of Theorem \ref{thm:fstthrm}]
Let $\beta >\beta_s$. By Lemma \ref{lem:ourplefkaprofs} 1) we have $0\not\in D_\beta$ so that from \eqref{eq:varprinc} and \eqref{eq: reformu}
\beq\label{eq: fTAP sup over E}
\fTAP(\beta) =\sup_{E \in [E_{\infty},E_0]} \left\{ f(E,q_E(E)) + I(E) \right\}
\eeq
Note that
\begin{equation}\label{eq: max over E}
    E \to f(E,q_E(E)) + I(E)
\end{equation} 
has exactly one critical point $E \ge E_{\infty}$ by \eqref{lem: crit eq sol}, \eqref{eq: useful ident} and \eqref{eq: qE glob max}. Also
$$ \frac{d}{d E}\left(f(E,\qE(E))+I(E)\right)|_{E=E_{\infty}} \overset{\eqref{eq: qE glob max}}{=} {\partial_E}f(E_{\infty},q_{\min}) + I'(E_{\infty}) \overset{\eqref{eq: deriv at Einf}}{>}0. 
$$
This shows that the unique maximizer $(\tilde{E},\tilde{q})$ of $f(E,q_E(E))+ I_{\Ann}(E)$ from Lemma \ref{lem:solveopti2} satisfies
$\tilde{E}>E_0$ when $\beta>\beta_s$ (recall \eqref{eq: E0 betas equiv}) the maximizer of \eqref{eq: max over E} in $[E_{\infty},E_0]$ must be $E=E_0$, and
$$\fTAP(\beta) = f(E_0, q_E(E_0)) + I(E_0) = f(E_0, q_E(E_0)) < f(\tilde{E},q_E(\tilde{E}))=\frac{\beta^2}{2}.$$
This proves \eqref{eq: FE low temp}.
By \eqref{eq: qE as sol conc} we have that $q_*=q_E(E_0)$ is the unique solution to \eqref{eq: qstar low temp}. By \eqref{eq: TOT TAP FE from f and I} and the fact the supremum in \eqref{eq: fTAP sup over E} is uniquely attained, so that \eqref{eq: def E TAP star} is well-defined and \eqref{eq: Vstar Ustar ITAPstar low temp} (a), (b) and \eqref{eq: FE low temp} follow. Since $f(E_0, q_E(E_0))=h_{\TAP}(q^{p/2}E_0,q)$ and using \eqref{eq: Umax as sup}-\eqref{eq: Umax form} we also obtain \eqref{eq: fTAP static low temp}.
\end{proof}
\printbibliography

\end{document}